\documentclass[journal,twocolumn,10pt]{IEEEtran}
\usepackage[T1]{fontenc}
\usepackage{cite}
\usepackage[cmex10]{amsmath}
\usepackage{amsthm}
\usepackage{amssymb}
\usepackage{mathtools}
\usepackage[hidelinks]{hyperref} 
\usepackage{algorithmic}
\usepackage{textcomp}
\usepackage[usenames, dvipsnames]{color}
\usepackage[linesnumbered,ruled]{algorithm2e}
\usepackage{subcaption}
\usepackage{graphicx}  
\usepackage{url}
\usepackage[usestackEOL]{stackengine}
\usepackage{verbatim}

\usepackage {tikz}
\tikzset{>=stealth}
\usepackage{xcolor,colortbl}
\definecolor{lightred}{RGB}{255, 240, 240}
\definecolor{lightblue}{RGB}{230, 250, 255}
\definecolor{lightgreen}{RGB}{240, 255, 242}
\definecolor{myred}{RGB}{220, 0, 0}
\definecolor{myblue}{RGB}{0, 17, 173}
\definecolor{mygreen}{RGB}{2, 117, 0}

\newtheorem{lemma}{Lemma}
\newtheorem{theorem}{Theorem}

\newcommand{\E}[1]{\mbb{E}{\left[#1\right]}}
\newcommand{\Var}[1]{\mrm{Var}{\left[#1\right]}}
\newcommand{\Cov}[1]{\mrm{Cov}{\left[#1\right]}}
\newcommand{\mca}[1]{\mathcal{#1}}

\newcommand{\mbf}[1]{\mathbf{#1}}
\newcommand{\mbb}[1]{\mathbb{#1}}
\newcommand{\mrm}[1]{\mathrm{#1}}
\newcommand{\mtt}[1]{\mathtt{#1}}
\newcommand{\bsm}[1]{\boldsymbol{#1}}

\newcommand{\tr}{\operatorname{tr}}
\newcommand{\rank}{\operatorname{rank}}

\newcommand{\diag}[1]{\operatorname{diag}{#1}}
\newcommand{\blkdiag}[1]{\operatorname{blkdiag}{#1}}

\DeclareMathOperator*{\argmin}{arg\;min}

\DeclareMathOperator*{\maximize}{maximize}

\newcommand{\Hd}{\mbf{H}_{\mrm{d}}}
\newcommand{\hd}[1]{\mbf{h}_{\mrm{d},{#1}}}

\newcommand{\Hck}[1]{\mbf{H}_{\mrm{c},{#1}}}
 
\newcommand{\mcA}{\mca{A}}
\newcommand{\mcB}{\mca{B}}
\newcommand{\RIS}{\mtt{RIS}}
\newcommand{\DDone}{\mtt{DD}1}
\newcommand{\PD}{\mtt{PD}}
\newcommand{\BS}{\mtt{BS}}
\newcommand{\taup}{\tau_{\mrm{p}}}
\newcommand{\tauc}{\tau_{\mrm{c}}}
\newcommand{\TI}{\mca{T}_{1}}
\newcommand{\TII}{\mca{T}_{2}}

\makeatletter
\renewcommand*\env@matrix[1][\arraystretch]{%
  \edef\arraystretch{#1}%
  \hskip -\arraycolsep
  \let\@ifnextchar\new@ifnextchar
  \array{*\c@MaxMatrixCols c}}
\makeatother

\ifCLASSINFOpdf
\else
\fi
\usepackage{amsmath}
\begin{document}
%
\title{Decision-Directed Hybrid RIS Channel Estimation with Minimal Pilot Overhead}
%
%
%

\author{
Ly~V.~Nguyen and A.~Lee~Swindlehurst
\thanks{Ly V. Nguyen and A. Lee Swindlehurst are with the Center for Pervasive Communications and Computing, Henry Samueli School of Engineering, University of California, Irvine, CA, USA 92697 (e-mail: vanln1@uci.edu, swindle@uci.edu).}
}

\maketitle

\begin{abstract}
To reap the benefits of reconfigurable intelligent surfaces (RIS), channel state information (CSI) is generally required. However, CSI acquisition in RIS systems is challenging and often results in very large pilot overhead, especially in unstructured channel environments. Consequently, the RIS channel estimation problem has attracted a lot of interest and also been a subject of intense study in recent years. In this paper, we propose a decision-directed RIS channel estimation framework for general unstructured channel models. The employed RIS contains some hybrid elements that can simultaneously reflect and sense the incoming signal. We show that with the help of the hybrid RIS elements, it is possible to accurately recover the CSI with a pilot overhead proportional to the number of users. Therefore, the proposed framework substantially improves the system spectral efficiency compared to systems with passive RIS arrays since the pilot overhead in passive RIS systems is proportional to the number of RIS elements times the number of users. We also perform a detailed spectral efficiency analysis for both the pilot-directed and decision-directed frameworks. Our analysis takes into account both the channel estimation and data detection errors at both the RIS and the BS. Finally, we present numerous simulation results to verify the accuracy of the analysis as well as to show the benefits of the proposed decision-directed framework.
\end{abstract}


%
\IEEEpeerreviewmaketitle

\section{Introduction}
\label{sec_introduction}
Reconfigurable intelligent surfaces (RIS) are a novel technology that has changed the conventional long-standing perspective that wireless channels are an uncontrollable part of the environment. RISs are planar arrays composed of elements whose electromagnetic reflection coefficients can be adaptively configured to shape the wireless channel in beneficial ways. As such, they can be deployed to improve the system throughput, network coverage, or energy efficiency~\cite{Basar2019Wireless,Marco2020Smart}. However, the exploitation of this channel-shaping ability generally requires RIS-related channel state information (CSI), which is challenging to obtain since the number of RIS elements can be very large, and the RIS elements are often constructed as passive devices without active radio-frequency (RF) chains or computational resources. Therefore, the RIS channel estimation problem has been a subject of intense study in the last few years~\cite{Swindlehurst2022GeneralFramework}. The literature of RIS channel estimation can be divided into two categories including structured and unstructured channel estimations. While structured channel estimation considers models that are parameterized by the angles of arrival (AoAs), angles of departure (AoDs), and complex gains of the propagation paths, unstructured channel estimation methods assume more generic channels described by arbitrary complex coefficients.

Numerous results on structured RIS channel estimation have been reported, for example in \cite{Wan2020Broadband,Ma2020Joint,Wang2020Compressed,He2021Atomic,Liu2021CascadedmmWave,Ardah2021Trice,Liu2021ADMM,Zhang2021Cascaded,Chen2022Hybrid,Wenhui2021Cascaded,Chen2023mmWave,Noh2022Training,Zhou2022Channel,Wei2021Double-Structured,Ma2021BeamSquint,Emil2022LoS,Wei2022Joint,Jianhe2023Semi}, where the sparsity property of high-frequency (e.g., millimeter-wave, or ``mmWave'') channels are exploited to reduce the pilot overhead. For example, the studies in \cite{Wan2020Broadband,Ma2020Joint,Wang2020Compressed,He2021Atomic,Liu2021CascadedmmWave,Ardah2021Trice,Liu2021ADMM,Zhang2021Cascaded,Chen2022Hybrid} formulated the cascaded mmWave channel estimation problem as a sparse signal recovery problem so that various compressive sensing techniques can be exploited to recover the channel parameters, e.g., distributed orthogonal matching pursuit (OMP)~\cite{ Wan2020Broadband }, iterative atom pruning based subspace pursuit (IAP-SP)~\cite{Ma2020Joint}, atomic norm minimization~\cite{He2021Atomic}, Newtonized orthogonal matching pursuit~\cite{Liu2021CascadedmmWave}, alternating direction method of multipliers (ADMM)~\cite{Liu2021ADMM}, and the hybrid multi-objective evolutionary paradigm~\cite{Chen2022Hybrid}. Several other system scenarios and designs were investigated in~\cite{ Wenhui2021Cascaded, Noh2022Training, Ma2021BeamSquint, Chen2023mmWave, Wei2021Double-Structured}. More specifically, the work in~\cite{Wenhui2021Cascaded} considers low-precision analog-to-digital converters (ADCs) at the BS and derives a linear channel estimator. The authors in~\cite{Noh2022Training} exploited the sparse structure of mmWave channels to derive a Cram{\'e}r-Rao lower bound (CRB) for the channel parameters, which is then optimized to design an RIS reflection pattern. The effect of beam squint was taken into account in~\cite{Ma2021BeamSquint} and a twin-stage orthogonal matching pursuit (TS-OMP) algorithm was developed to estimate the channel parameters. The double-structured angular sparsity of cascaded channels was exploited in~\cite{Chen2023mmWave, Wei2021Double-Structured} to both reduce the pilot overhead and improve the estimation performance. The work in~\cite{Emil2022LoS} developed a maximum likelihood (ML) channel estimation framework for estimating the line-of-sight (LoS) user-RIS channel. Exploiting the fact that the channel angles vary much slower than the channel gains, the authors in~\cite{Zhou2022Channel} proposed a two-timescale parametric estimation strategy which estimates all the channel angles and gains in the first coherence block, and then only re-estimates the channel gains in the remaining coherence blocks. Joint channel estimation and data detection taking into account the sparsity and low-rank structure of mmWave channels was studied in~\cite{Wei2022Joint,Jianhe2023Semi} where Taylor series expansion and Gaussian approximation were used in~\cite{Wei2022Joint} and a two-stage fitting algorithm was derived~\cite{Jianhe2023Semi}.

Unlike the aforementioned works where all the RIS elements are assumed to be passive, some other structured channel estimation studies in~\cite{Liu2020Deep,Taha2021Enabling,Chen2021LowComplexity,Jin2021DRN,Hu2022SemiPassive,Kim2023Semi-Passive} assume that the RIS contains a small number of active elements that can operate in sensing mode to estimate partial CSI, which is then exploited together with the sparsity structure of mmWave channels to reconstruct the full CSI. While compressed sensing methods were used in~\cite{Liu2020Deep,Taha2021Enabling}, some other techniques were employed in~\cite{Chen2021LowComplexity,Jin2021DRN,Hu2022SemiPassive}, e.g., signal parameters via rotational invariance technique (ESPRIT) and multiple signal classification (MUSIC) in~\cite{Chen2021LowComplexity,Hu2022SemiPassive} and deep residual networks in~\cite{Jin2021DRN}. Unlike the methods in~\cite{Liu2020Deep,Taha2021Enabling,Chen2021LowComplexity,Jin2021DRN,Hu2022SemiPassive} that require both uplink and downlink training signals, the work in~\cite{Kim2023Semi-Passive} developed a variational inference-sparse Bayesian learning channel estimator that uses only the uplink training signals and exploits the received signals at both the RIS and the BS.

On the other hand, unstructured RIS channel estimation has also been rigorously investigated in many works, e.g., literature~\cite{Mishra2019Channel,Yang2020Intelligent,Jensen2020Optimal,Zheng2020Intelligent,Nadeem2020Intelligent,You2020Channel,Demir2022Is,Gilderlan2021Channel,Liu2020Matrix,Wei2021ChannelMISO,He2020Low-Rank,Wang2020Channel,Wei2022Channel}. For single-user systems, the works in~\cite{Mishra2019Channel,Yang2020Intelligent} used a binary reflection strategy where only one reflecting element is turned on in each time slot. It was then shown in~\cite{Jensen2020Optimal, Zheng2020Intelligent} that turning on all the RIS elements at the same time and using a discrete Fourier transform (DFT) matrix as the reflecting pattern provides better performance compared to the binary reflection strategy. Similar results were also reported for the case of multiple users in~\cite{Nadeem2020Intelligent}. Additionally, the study in~\cite{You2020Channel} examines the reflecting pattern design problem while imposing the restriction that the phase shifts are limited to a finite set of discrete values. For multi-user systems, the work in~\cite{Demir2022Is} exploits known spatial correlation at both the BS and the RIS as well as other statistical characteristics of multi-specular fading to derive Bayesian channel estimators. The work in~\cite{ He2020Low-Rank} assumes a low-rank RIS-BS channel and develops a two-stage algorithm based on matrix factorization and matrix completion. Some other methods such as matrix-calibration-based factorization and parallel factor tensor decomposition were used in \cite{Liu2020Matrix} and~\cite{Gilderlan2021Channel,Wei2021ChannelMISO}, respectively. More general channel models were considered in~\cite{Wei2022Channel,Wang2020Channel} where two- and three-phase estimation approaches were proposed, respectively. While both of these latter approaches require the same pilot overhead, the two-phase approach outperforms the other thanks to the alleviation of error propagation. Joint channel estimation and data detection for unstructured channels were also studied in~\cite{Huang2022Semi,Alwakeel2022Semi}. Expectation-maximization was exploited in ~\cite{Huang2022Semi} for a single-user system, and an on/off reflection strategy was used in~\cite{Alwakeel2022Semi} for a multiuser systems.

In all of the above unstructured channel estimation methods, passive RIS arrays were used. In this paper, we consider a recent hybrid RIS structure~\cite{Alexandropoulos2021Hybrid,Alamzadeh2021Reconfigurable,Zhang2022Channel} in which the RIS elements can simultaneously reflect and sense the incoming signal, and develop a decision-directed (DD) channel estimator that can be used for unstructured channels where AoA/AoD information cannot be exploited. The novelty of the approach lies in the application of hybrid RIS for unstructured channel estimation, and the use of DD to reduce the pilot overhead. The contributions of this paper are summarized as follows:
\begin{itemize}
    \item Based on the hybrid RIS structure, we first develop a two-phase pilot-directed (PD) channel estimation approach. The estimation strategy is similar to that in~\cite{Wei2022Channel} but we show that the pilot overhead is lower for multiuser systems.
    \item Next, we propose a two-phase DD channel estimation framework and we show that with the help of the hybrid RIS elements, it is possible to accurately recover the CSI with a pilot overhead only proportional to the number of users. Therefore, the proposed DD framework substantially improves the system spectral efficiency (SE). More specifically, in the channel estimation stage, the users transmit a sequence including both pilot and data symbols where the number of pilot symbols is the same as the number of users. The RIS uses some sensor elements with RF chains to recover the data symbols, and then forwards the detected data symbols to the BS for cascaded channel estimation. For the BS to accurately estimate the CSI, the RIS phase shifts must be varied. We point out that changing the RIS phase shifts does not affect data detection by the sensing RIS elements, and thus both data recovery at the RIS and channel estimation at the BS are guaranteed. We also explain why accurate CSI recovery is not guaranteed when the DD approach is applied at the BS and the RIS has no sensing elements.
    \item We then perform a detailed spectral efficiency (SE) analysis for both the PD and DD frameworks for single-user systems. Our analysis takes into account both the channel estimation and data detection errors at the RIS and the BS, and thus accurately reflects the uncertainty of RIS-assisted data detection in the DD framework. It is observed that there is often a crossing point at which the DD framework outperforms the PD one, and so the analysis can be used to decide when the PD or DD approach should be used. Finally, we present numerous simulation results to verify the accuracy of the SE analysis as well as to show the benefits of the proposed DD framework.
\end{itemize}

The rest of this paper is organized as follows: Section~\ref{sec_system_model} presents the considered system model. The pilot directed and decision-directed channel estimation frameworks are presented in Section~\ref{sec_PD} and Section~\ref{sec_DD}, respectively. We perform the spectral efficiency analysis in Section~\ref{sec_SE}. Section~\ref{sec_numerical_results} shows simulation results and finally Section~\ref{sec_conclusion} concludes the paper.

\textit{Notation:} Upper-case and lower-case boldface letters denote matrices and column vectors, respectively. Scalars $x_{ij}$ and $[\mbf{X}]_{ij}$ both denote the element at the $i$th row and $j$th column of a matrix $\mbf{X}$. Vectors $\mbf{x}_i$ and $\mbf{X}_{:,i}$ both denote the $i$th column of a matrix $\mbf{X}$, while $\mbf{X}_{k,:}$ denotes the $k$-th row of $\mbf{X}$. The notation $\mbf{X}_{i:j,k:\ell}$ represents the sub-matrix of $\mbf{X}$ that includes rows $i$ to $j$ and columns $k$ to $\ell$. The expectation, variance, and covariance of random quantities are denoted by $\mathbb{E}[\cdot]$, $\Var{\cdot}$, and $\Cov{\cdot}$, respectively. Depending on the context, the operator $|\cdot|$ is used to denote the absolute value of a number, or the cardinality of a set. The $\ell_2$-norm of a vector is represented by $\|\cdot\|$. The transpose and conjugate transpose are denoted by $[\cdot]^T$ and $[\cdot]^H$, respectively, $j$ is the unit imaginary number satisfying $j^2=-1$, $\mathcal{N}(\cdot,\cdot)$ and $\mathcal{CN}(\cdot,\cdot)$ represent the real and the complex normal distributions respectively, where the first argument is the mean and the second argument is the variance or the covariance matrix. The $i$-th element of the set $\mcA$ is indicated by $\mcA(i)$. The Q-function that quantifies the tail distribution of a standard normal random variable is given by~$Q(\cdot)$.

\section{System Model}
\label{sec_system_model}
We consider an uplink RIS-assisted MIMO system in which a BS with $M$ antennas serves $K$ single-antenna users under the assistance of an $N$-element RIS. Let $\Hd \in \mbb{C}^{M\times K}$, $\mbf{H} \in \mbb{C}^{M\times N}$, and $\mbf{G} \in \mbb{C}^{N\times K}$ denote the direct channel from the users to the BS, the channel from the RIS to the BS, and the channel from the users to the RIS, respectively. The RIS contains a number of sensing elements equipped with radio-frequency (RF) chains as illustrated in Fig.~\ref{fig:system_model}. These sensing elements are able to simultaneously reflect and sense the impinging signal. Let $\mcA$ denote the index set of the sensing elements, so that $\mcA \subset \{1,\,\ldots,\,N\}$, and let $N_\mcA$ be the number of sensing elements, i.e., $N_\mcA = |\mcA|$, where it is assumed that $K \leq N_\mcA \ll N$. 

Define the channel matrices $\Hd \overset{\Delta}{=} [\hd{1},\,\ldots,\,\hd{K}]$ and $\mbf{G} \overset{\Delta}{=} [\mbf{g}_1,\,\ldots,\,\mbf{g}_K]$, so that the received signal at the BS is modeled as
\begin{align}
    \mbf{y}^\BS 
    &= \sqrt{P}\sum_{k=1}^K (\hd{k} +\mbf{H}\diag(\mbf{g}_k)\diag(\bsm{\rho})\bsm{\phi}) s_k + \mbf{n}^\BS\\
    &= \sqrt{P}\sum_{k=1}^K \Hck{k} \diag{\left(\begin{bmatrix}[0.8]
        1\\ \bsm{\rho}
    \end{bmatrix}\right)} \begin{bmatrix}[0.8]
        1\\ \bsm{\phi}
    \end{bmatrix} s_k + \mbf{n}^\BS
\end{align}
where $\bsm{\phi} = [\phi_1,\, \ldots,\,\phi_N]^T$ is the phase shift vector of the RIS, $\Hck{k} = [\hd{k},\, \mbf{H}\diag(\mbf{g}_k)] \in \mbb{C}^{M\times (N+1)}$ is the cascaded channel of the $k$-th user, $P$ is the transmit power, and $\bsm{\rho} \overset{\Delta}{=} [\rho_1,\,\ldots,\,\rho_N]^T$ with $0 \leq \rho_n \leq 1$ if $n \in \mcA$, otherwise $\rho_n = 1$. Hence, $\rho^2_n$ is the portion of the power of the impinging signal that is reflected by the $n$-th RIS element. For convenience, we use the notation $\bsm{\rho}^\mcA = [\rho_1^\mcA,\,\ldots,\,\rho_{N_\mcA}^\mcA]^T$  where $\rho_i^\mcA \overset{\Delta}{=} \rho_{\mcA(i)}$ for $i = 1,\,\ldots,\,N_\mcA$, and $\bsm{\eta}^\mcA = [\eta_1^\mcA,\,\ldots,\,\eta_{N_\mcA}^\mcA]^T$ where $\eta_i^\mcA = \sqrt{1 - (\rho^\mcA_i)^2}$. Hence, $(\eta_i^\mcA)^2$ represents the amount of signal power absorbed by the RIS element $\mca{A}(i)$.

With $N_\mcA$ sensing elements at the RIS, the received signal at the RIS is given as
\begin{equation}
    \mbf{{y}^\RIS} = \sqrt{P} \diag(\bsm{\eta}^\mcA) \diag( \bsm{\phi}^\mcA) \sum_{k=1}^K \mbf{g}^\mcA_{k} s_{k} + \mbf{n}^\RIS \; ,
\end{equation}
where $\mbf{g}_k^\mcA = [g_{k,1}^\mcA,\,\ldots,\,g_{k,N_\mcA}^\mcA]^T$ with $g^\mcA_{k,i} \overset{\Delta}{=} g_{k,\mcA(i)}$. In this paper, it is important to note that the superscripts $(\cdot)^\mcA$ and $(\cdot)^\mcB$ are used to imply variables that are associated with the sensing and reflecting RIS elements, respectively. 

We assume an uplink communication protocol with two stages including a channel estimation stage followed a data transmission stage. After the channel estimation stage, the RIS phase shifts are optimized and configured before the data transmission stage begins. It should be noted that during the channel estimation stage, data detection occurs at the RIS since the users transmit both pilot and data symbols during this stage. During the data detection stage, in order to minimize power consumption at the RIS, the sensing function of the hybrid RIS elements is turned off and the incoming signal is completely reflected.

\begin{figure}
    \centering
    \includegraphics[width=\linewidth]{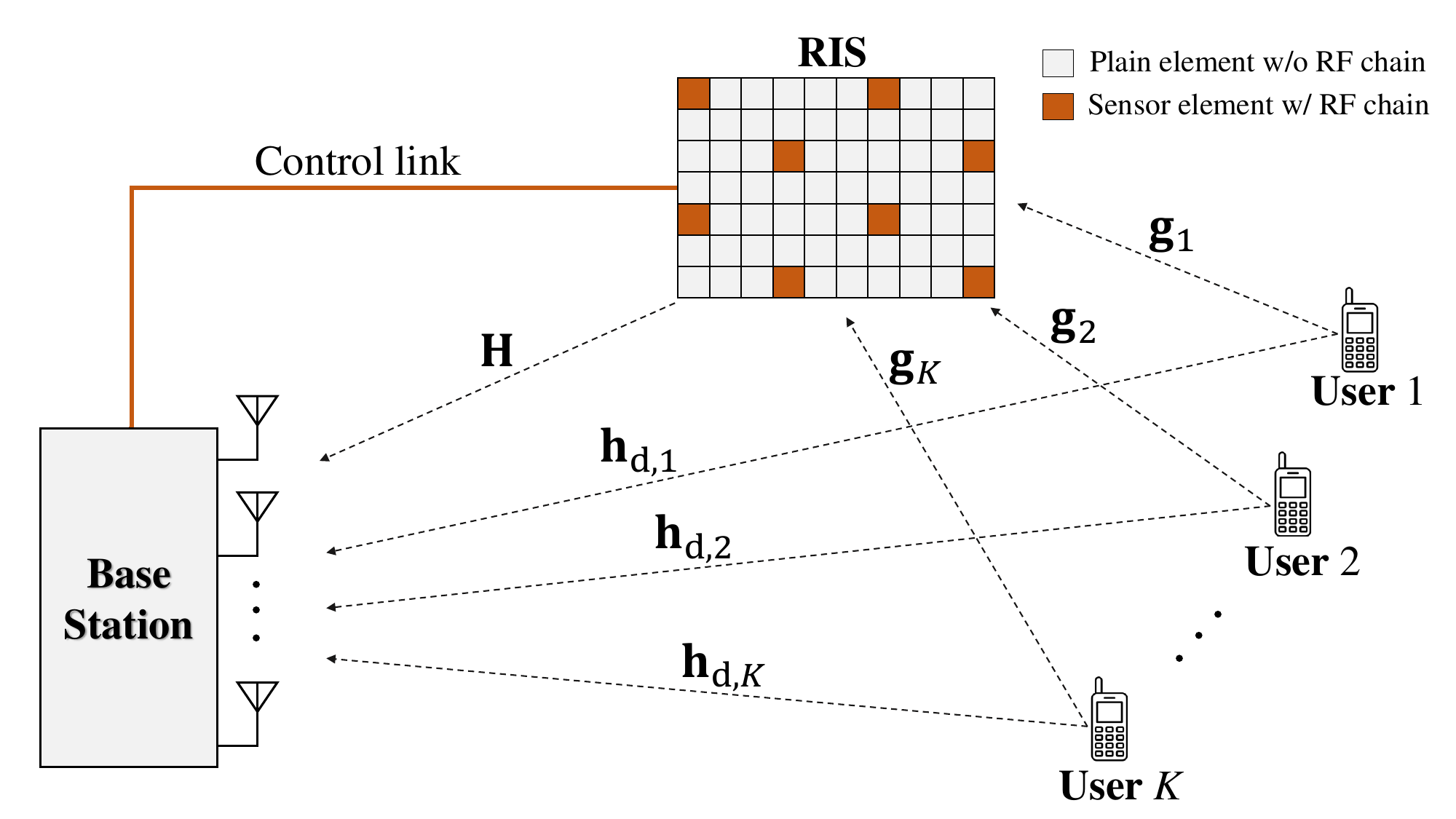}
    \caption{Sensing-RIS-assisted multi-user MIMO system.}
    \label{fig:system_model}
\end{figure}
\section{Pilot-Directed Channel Estimation}
\label{sec_PD}
In this section, we present a two-phase pilot-directed approach for estimating the cascaded channel matrices $\Hck{1},\, \ldots,\, \Hck{K}$. Since all the users experience the same RIS-BS channel, i.e. the same channel matrix $\mbf{H}$, the total number of channel elements to be estimated is $M(K + N) + N(K-1)$~\cite{Wang2020Channel,Wei2022Channel}. Let $\mbf{A}_k = \mbf{H}\diag(\mbf{g}_k)$, then we have $\mbf{A}_k = \mbf{A}_1\diag(\bsm{\lambda}_k)$ where $\lambda_{k,n} = g_{k,n}/g_{1,n}$ for $k = 2,\,\ldots,\,K$ and $n = 1,\,\ldots,\,N$. Note that $\bsm{\lambda}_1 = \mbf{1}_N$. Therefore, it suffices to estimate $\Hd$, $\mbf{A}_1$, and $\bsm{\lambda}_2,\, \ldots,\,\bsm{\lambda}_K$. 


Our two-phase estimation strategy is similar to that in~\cite{Wei2022Channel} where $\Hck{1} = [\hd{1},\, \mbf{A}_1]$ is estimated in phase~1 and $\hd{2},\, \ldots,\, \hd{K}$ and $\bsm{\lambda}_2,\, \ldots,\, \bsm{\lambda}_K$ are estimated in phase~2. However, unlike the work in~\cite{Wei2022Channel}, which considers an RIS with passive elements only, our work here considers a hybrid RIS structure as presented above. In this section, we assume that only pilot signals are used for the channel estimation. For notational convenience, let $\TI = \{1,\,\ldots,\,\tau_1\}$ and $\TII = \{\tau_1+1,\,\ldots,\,\tau_1+\tau_2\}$ where $\tau_1$ and $\tau_2$ are the length of phase~1 and phase~2, respectively.

\subsection{Phase 1}
In this phase, we estimate $\Hck{1} = [\hd{1},\,\mbf{A}_1]$. One selected user transmits a pilot vector of length $N+1$, while the other users remain idle. Without loss of generality, we set the index of the typical user to~1. The received signal at the BS in this phase is given as
\begin{equation}
    \mbf{y}_t^\BS = \sqrt{P}\Hck{1}\diag{\left(\begin{bmatrix}[0.8]
        1\\ \bsm{\rho}
    \end{bmatrix}\right)} \begin{bmatrix}[0.8]
        1\\ \bsm{\phi}_t
    \end{bmatrix} s_{1,t} + \mbf{n}_t^\BS.
\end{equation}
Since $\Hck{1}$ contains $N+1$ columns, we need at least $\tau_1 = N+1$ time slots to accurately estimate $\Hck{1}$. For simplicity, we can set the pilot vector as $\mbf{S}_{1,\TI}  = \mbf{1}_{\tau_1}^T$ and the RIS phase shift matrix $\bsm{\Phi}$ is chosen so that $[\mbf{1}_{\tau_1},\bsm{\Phi}^T]^T = \mbf{V}_{\tau_1}$ where $\mbf{V}_{\tau_1}$ is the DFT matrix of size $\tau_1\times \tau_1$. This means $[1,\bsm{\phi}_t^T]^T$ is the $t$-th column of $\mbf{V}_{\tau_1}$. Then, the cascaded channel $\mbf{H}_{\mrm{c},1}$ can be estimated via standard methods, such as for example the least-squares (LS):
\begin{align}
    \mbf{\hat{H}}_{\mrm{c},1} &= \frac{1}{\sqrt{P}\tau_1}\mbf{Y}^\BS_{:,\TI}\bsm{\Phi}^H_{\tau_1}\diag{\left(\begin{bmatrix}[0.8]
        1\\ \bsm{\rho}
    \end{bmatrix}\right)}^{-1} \notag \\
    &= \mbf{{H}}_{\mrm{c},1} + \frac{1}{\sqrt{P}\tau_1} \mbf{N}_{:,\mca{T}_1}^\BS\bsm{\Phi}^H_{\tau_1}\diag{\left(\begin{bmatrix}[0.8]
        1\\ \bsm{\rho}
    \end{bmatrix}\right)}^{-1}.
\end{align}

The received signal at the sensing elements of the RIS is
\begin{equation}
    \mbf{y}^\RIS_t = \sqrt{P} \diag(\bsm{\eta}^\mcA) \diag( \bsm{\phi}^\mcA_t) \mbf{g}^\mcA_1 + \mbf{n}^\RIS_t
    \label{eq:rx_signal_RIS_phase1}
\end{equation}
and so the sensed portion of the channel $\mbf{g}^\mcA_1$ can be estimated as
\begin{align}
    \mbf{\hat{g}}^\mcA_1 
    &=\frac{1}{\sqrt{P}\tau_1}\diag(\bsm{\eta}^\mcA)^{-1} \sum_{t=1}^{\tau_1} \bsm{\psi}_t.\label{eq:g1A_hat}
\end{align}
where $\bsm{\psi}_t = \diag(\bsm{\phi}_t^\mcA)^{-1}\mbf{y}^\RIS_t$.

\subsection{Phase 2}
In this phase, the typical user remains idle while the other users transmit pilot sequences.  Since $\hd{1}$ and $\mbf{A}_1$ have been estimated in phase 1, we will estimate $\hd{2},\, \ldots,\, \hd{K}$ and $\bsm{\lambda}_2,\, \ldots,\, \bsm{\lambda}_K$ during phase 2. The received signal at the BS can be decomposed as
\begin{align}
    \mbf{y}_t^\BS &= \sqrt{P}\sum_{k=2}^K \big(\hd{k} +\mbf{A}_1^{\mcB}\diag(\bsm{\phi}_t^{\mcB})\bsm{\lambda}_k^{\mcB} \; + \notag \\ &\hspace{1.5cm} \mbf{A}_1^{\mcA}\diag(\bsm{\rho}^{\mcA})\diag(\bsm{\phi}_t^{\mcA})\bsm{\lambda}_k^{\mcA}\big) s_{k,t} + \mbf{n}_t^\BS
    \label{eq:rx_BS_phase2}
\end{align}
where $\mbf{A}_1^{\mcA}$ and $\mbf{A}_1^{\mcB}$ are matrices whose columns are drawn from $\mbf{A}_1$ with indices $\mcA$ and $\mcB$, respectively, $\bsm{\lambda}_k^{\mcA} = [\lambda_{k,1}^{\mcA},\,\ldots,\,\lambda_{k,N_\mcA}^{\mcA}]^T$ and $\bsm{\lambda}_k^{\mcB} = [\lambda_{k,1}^{\mcB},\,\ldots,\,\lambda_{k,N_\mcB}^{\mcB}]^T$ where $\lambda_{k,i}^{\mcA} \overset{\Delta}{=} \lambda_{k,\mcA(i)}$ and $\lambda_{k,i}^{\mcB} \overset{\Delta}{=} \lambda_{k,\mcB(i)}$.

\subsubsection{Estimating $\bsm{\lambda}_2^\mcA,\, \ldots,\, \bsm{\lambda}_K^\mcA$} This is done at the RIS. Since $\lambda^\mcA_{k,i} = g_{k,i}^\mcA / g_{1,i}^\mcA$, the parameters $\bsm{\lambda}_2^\mcA,\, \ldots,\, \bsm{\lambda}_K^\mcA$ are defined as long as $\mbf{G}^\mcA = [\mbf{g}_1^\mcA,\,\ldots,\,\mbf{g}_K^\mcA]$ is known. Note that the first column of $\mbf{G}^\mcA$ has been estimated using~\eqref{eq:g1A_hat} in phase 1. The signal received at the RIS in phase 2 is given as
\begin{align}
    \mbf{y}^\RIS_t 
    &= \sqrt{P}\diag(\bsm{\eta}^\mcA)\diag(\bsm{\phi}^\mcA_t)\mbf{G}^\mcA_{:,2:K}\mbf{S}_{2:K,t} + \mbf{n}^\RIS_t.
\end{align}
The sub-matrix $\mbf{G}^\mcA_{:,2:K}$ can be estimated by the RIS as follows:
\begin{equation}
    \mbf{\hat{G}}^\mcA_{:,2:K} = \frac{1}{\sqrt{P}}\diag(\bsm{\eta}^\mcA)^{-1}\bsm{\Psi}_{\TII}\mbf{S}_{2:K,\TII}^H(\mbf{S}_{2:K,\TII}\mbf{S}_{2:K,\TII}^H)^{-1}
\end{equation}
where $\bsm{\Psi}_{\TII} = [\bsm{\psi}_{\tau_1+1},\,\ldots,\,\bsm{\psi}_{\tau_1+\tau_2}]$. Thus, an estimate of ${\lambda}^\mcA_{k,i}$ can be obtained as
$\hat{\lambda}^\mcA_{k,i} = {\hat{g}_{k,i}^\mcA}/{\hat{g}_{1,i}^\mcA}$.

\subsubsection{Estimating $\hd{2},\, \ldots,\, \hd{K}$ and $\bsm{\lambda}^\mcB_2,\, \ldots,\, \bsm{\lambda}^\mcB_K$} This is accomplished at the BS. Let 
\begin{align*}
    \mbf{B}_t &= [\mbf{I}_M,\, \mbf{A}_1^\mcB\diag(\bsm{\phi}^\mcB_t)],\\
    \bsm{\upsilon}_k &= [\hd{k}^T,\, (\bsm{\lambda}_k^\mcB)^T]^T,\\
    \mbf{f}_{k,t}^\mcA &= \mbf{A}_1^{\mcA}\diag(\bsm{\rho}^{\mcA})\diag(\bsm{\phi}_t^{\mcA})\bsm{\lambda}_k^{\mcA},
\end{align*}
then the received signal at the BS in~\eqref{eq:rx_BS_phase2} can be written in the following form:
\begin{align}
    \mbf{y}_t^\BS &= \sqrt{P}\sum_{k=2}^K (\mbf{B}_t \bsm{\upsilon}_k + \mbf{f}_{k,t}^\mcA) s_{k,t} + \mbf{n}_t^\BS\notag\\
    &= \sqrt{P}((\mbf{S}_{2:K,t}^T\otimes\mbf{B}_t)\bsm{\upsilon} + \mbf{F}_t^\mcA\mbf{S}_{2:K,t}) + \mbf{n}_t^\BS\notag\\
    &= \sqrt{P}(\mbf{Q}_t\bsm{\upsilon} + \mbf{\tilde{y}}^\BS_t) + \mbf{n}_t^\BS
    \label{eq:rx_BS_phase2_rewrite}
\end{align}
where $\bsm{\upsilon} = [\bsm{\upsilon}_2^T,\,\ldots,\,\bsm{\upsilon}_K^T]^T$, $\mbf{Q}_t = \mbf{S}_{2:K,t}^T\otimes\mbf{B}_t$, $\mbf{F}_t^\mcA = [\mbf{f}_{2,t}^\mcA,\,\ldots,\,\mbf{f}_{K,t}^\mcA]$, and $\mbf{\tilde{y}}^\BS_t = \mbf{F}_t^\mcA\mbf{S}_{2:K,t}$. Stacking the received signals $\{\mbf{y}_t^\BS\}$ in~\eqref{eq:rx_BS_phase2_rewrite} with $t \in \mca{T}_2$ on top of each other, we have the following 
\begin{equation}
    \mrm{vec}\left(\mbf{Y}^\BS_{:,\TII}\right) - \sqrt{P}\,\mrm{vec}\left(\mbf{\tilde{Y}}^\BS_{:,\TII}\right)= \sqrt{P} \mbf{Q}\bsm{\upsilon} + \mbf{n}^\BS
\end{equation}
where $\mbf{Q} = [\mbf{Q}_{\tau_1+1}^T,\,\ldots,\,\mbf{Q}_{\tau_1+\tau_2}^T]^T$.
Note that $\bsm{\upsilon}$ is the vector we need to estimate and the size of $\mbf{Q}$ is $M\tau_2\times(K-1)(M+N-N_\mcA)$. Therefore, in order to accurately recover $\bsm{\upsilon}$, two conditions should be satisfied: $M\tau_2 \geq (K-1)(M+N-N_\mcA)$ and $\mrm{rank}(\mbf{Q}) = M+N-N_\mcA$. An estimate of $\bsm{\upsilon}$ can be then obtained as
\begin{equation}
    \bsm{\hat{\upsilon}} = \frac{1}{\sqrt{P}}\mbf{Q}^{\dagger}\left(\mrm{vec}\left(\mbf{Y}^\BS_{:,\TII}\right) - \sqrt{P}\,\mrm{vec}\left(\mbf{\tilde{Y}}^\BS_{:,\TII}\right)\right).\label{eq:upsilon_hat}
\end{equation} 
If $M \geq N - N_\mcA$, we need at least $\tau_2 = 2(K-1)$ time slots to recover $\bsm{\upsilon}$ and if $M < N - N_\mcA$, we need at least $\tau_2 = K - 1 + \left \lceil \frac{(K-1)(N-N_\mcA)}{M} \right \rceil$ time slots to recover $\bsm{\upsilon}$.

When $M \geq N - N_\mcA$, user-$k$ transmits in two time slots with $\tau_1 + 2(k-2)+1$ and $\tau_1 + 2(k-2)+2$, and the estimation solution in~\eqref{eq:upsilon_hat} can be decomposed into $K-1$ separate expressions as follows:
\begin{align}
    \bsm{\hat{\upsilon}}_k &= \frac{1}{\sqrt{P}}\begin{bmatrix}[1.15]
        s_{k,\tau_1+2(k-2)+1} \mbf{B}_{\tau_1+2(k-2)+1} \\
        s_{k,\tau_1+2(k-2)+2} \mbf{B}_{\tau_1+2(k-2)+2} 
    \end{bmatrix}^{\dagger}\times \notag\\
    &\qquad\left(\begin{bmatrix}[1.15]
        \mbf{y}^\BS_{\tau_1+2(k-2)+1}\\
        \mbf{y}^\BS_{\tau_1+2(k-2)+2}
    \end{bmatrix} - \sqrt{P}\begin{bmatrix}[1.15]
        \mbf{\tilde{y}}^\BS_{\tau_1+2(k-2)+1}\\
        \mbf{\tilde{y}}^\BS_{\tau_1+2(k-2)+2}
    \end{bmatrix}\right).
\end{align}

\subsection{Overall Training Overhead}
Since the pilot overhead in phase 1 is $\tau_1 = N + 1$ and in phase 2 is $\tau_2 = 2(K-1)$ if $M \geq N - N_\mcA$, or $\tau_2 = K - 1 + \left \lceil \frac{(K-1)(N-N_\mcA)}{M} \right \rceil$ otherwise, the total pilot overhead will be
\begin{align}
    \taup &= \tau_1 + \tau_2 \notag \\
    &= \begin{cases}
        N+2K-1, &\text{if} \; M \geq N - N_\mcA,\\
        N+K+\left\lceil\frac{(K-1)(N-N_\mcA)}{M}\right\rceil, &\text{otherwise.}
    \end{cases}
    \label{eq:taup_PD2}
\end{align}
Although the total pilot overhead is less than $K(N+1)$ as required by the generic LS channel estimator at the BS, it is still proportional to the number of RIS elements $N$, which can be excessively large. In the following section, we propose a DD channel estimation approach whose pilot overhead is only proportional to the number of users $K$.

\section{Decision-Directed Channel Estimation}
\label{sec_DD}
In this section, we propose a two-stage DD channel estimation approach to substantially reduce the pilot overhead, and thus improve the system SE. As in the PD approach above, there are also two phases as detailed in what follows.

\subsection{Phase 1}
Similar to the PD approach, here we also use $\tau_1 = N+1$ time slots to estimate $\Hck{1}$ with the typical user active and the other users remaining idle. However, unlike the PD approach where all of the $N+1$ time slots are used for pilot signals, the DD approach uses only the first time slot to transmit a pilot symbol and the remaining $N$ time slots are for transmitting data symbols.

\subsubsection{Estimating $\mbf{g}^\mcA_1$}
The received signal at the RIS in the first time slot with the pilot symbol $s_{1,1}=1$ is given in~\eqref{eq:rx_signal_RIS_phase1}, and so an estimate of $\mbf{g}^\mcA_1$ can be obtained as
\begin{equation}
    \mbf{\hat{g}}^\mcA_1 = \frac{1}{\sqrt{P}} \diag(\bsm{\eta}^\mcA)^{-1} \diag( \bsm{\phi}^\mcA_1)^{-1} \mbf{y}^{\RIS}_1.
    \label{eq:g1A_hat_DD}
\end{equation}

\subsubsection{Data Detection for user 1 in time slots $2, \, \ldots,\, \tau_1$} The received signal is $\mbf{y}^\RIS_t = \sqrt{P}\diag(\bsm{\eta}^\mcA) \diag(\bsm{\phi}^\mcA_t) \mbf{{g}}^\mcA_1 s_{1,t} + n_t^\RIS$. Therefore, the data symbols $s_{1,t}$ can be detected by the RIS using $\mbf{\hat{g}}^\mcA_1$ from~\eqref{eq:g1A_hat_DD} as follows:
\begin{equation}
    \hat{s}_{1,t} = \argmin_{s \in \mca{S}} \big\|\mbf{y}^\RIS_t - \sqrt{P}\diag(\bsm{\eta}^\mcA) \diag(\bsm{\phi}^\mcA_t) \mbf{\hat{g}}^\mcA_1 s\big\|^2.
    \label{eq:s_hat_phase1}
\end{equation}
It can be seen that even when the phase shift vector $\bsm{\phi}^\mcA_t$ varies in different time slots, it is still possible for the RIS to accurately recover the data symbols since the effect of $\bsm{\phi}^\mcA_t$ is merely a phase rotation of the noiseless received signal which can be easily taken into account as in~\eqref{eq:s_hat_phase1}.

\subsubsection{Estimating $\Hck{1}$}
The detected data symbols $\{\hat{s}_{1,t}\}$ in~\eqref{eq:s_hat_phase1} will be forwarded by the RIS to the BS to estimate the cascaded channel matrix $\mbf{H}_{\mrm{c},1}$ as follows:
\begin{equation}
    \mbf{\hat{H}}_{\mrm{c},1} = \frac{1}{\sqrt{P}\tau_1}\mbf{Y}^\BS_{:,\TI}\diag{\left(\mbf{\hat{s}}_{\TI}\right)}^{-1}\bsm{\Phi}^H_{\tau_1}\diag{\left(\begin{bmatrix}[0.8]
        1\\ \bsm{\rho}
    \end{bmatrix}\right)}^{-1}.
\end{equation}
where $\mbf{\hat{s}}_{\TI} = [1,\, \hat{s}_{1,2},\, \ldots,\, \hat{{s}}_{1,\tau_1}]$.
Thus, in phase 1, with the help of the sensing elements, we use only one time slot for pilot signalling, i.e., $\tau_1 = 1$, while the other $N$ time slots are used for data transmission. The BS is still able to accurately recover the cascaded channel matrix $\Hck{1}$ as long as the data symbols are correctly detected by the RIS.

\subsection{Phase 2}
We divide phase 2 into two sub-phases that we refer to as 2a and 2b. Sub-phase 2a is associated with the time frame $\mca{T}_{\mrm{2a}} = \tau_1+1,\,\ldots,\,\tau_1+K-1$ where user 1 is idle and users $2$ through $K$ transmit their pilot signals. Sub-phase 2b is associated with the time frame $\mca{T}_{\mrm{2b}} = \tau_1+K,\,\ldots,\,\tau_1+\tau_2$ where all the users transmit data symbols.

\subsubsection{Estimating $\mbf{g}^\mcA_2,\,\ldots,\,\mbf{g}^\mcA_K$ and $\bsm{\lambda}_2^\mcA,\, \ldots,\, \bsm{\lambda}_K^\mcA$}
Pilot signals are transmitted in the first $K-1$ time slots of phase 2, from $\tau_1 + 1$ to $\tau_1 + K-1$, so the sub-matrix $\mbf{G}^\mcA_{:,2:K} = [\mbf{g}^\mcA_2,\,\ldots,\,\mbf{g}^\mcA_K]$ can be estimated by the RIS as follows:
\begin{equation}
    \mbf{\hat{G}}^\mcA_{:,2:K} = \frac{\diag(\bsm{\eta}^\mcA)^{-1}\bsm{\Psi}_{\mca{T}_{\mrm{2a}}}\mbf{S}_{2:K,\mca{T}_{\mrm{2a}}}^H(\mbf{S}_{2:K,\mca{T}_{\mrm{2a}}}\mbf{S}_{2:K,\mca{T}_{\mrm{2a}}}^H)^{-1}}{\sqrt{P}},
\end{equation}
where $\bsm{\Psi}_{\mca{T}_{\mrm{2a}}} = [\bsm{\psi}_{\tau_1+1},\,\ldots,\,\bsm{\psi}_{\tau_1+K-1}]$. Furthermore, an estimate of ${\lambda}^\mcA_{k,i}$ can be obtained as
$\hat{\lambda}^\mcA_{k,i} = {\hat{g}_{k,i}^\mcA}/{\hat{g}_{1,i}^\mcA}$.

\subsubsection{Detecting data}
For the remaining time slots from $\tau_1+K$ to $\tau_1 + \tau_2$ in sub-phase 2b, all $K$ users can transmit data, and the received signal at the RIS is
\begin{equation}
    \mbf{{y}}^\RIS_t = \sqrt{P} \diag(\bsm{\eta}^\mcA) \diag( \bsm{\phi}^\mcA_t) \mbf{G}^\mcA \mbf{s}_{t} + \mbf{n}^\RIS_t.
\end{equation}
The RIS can use $\mbf{y}^\RIS_t$ and $\mbf{\hat{G}}^\mcA = [\mbf{\hat{g}}^\mcA_1,\, \ldots,\, \mbf{\hat{g}}^\mcA_K]$ to detect the users' transmitted data $\mbf{s}_t$, which is a conventional MIMO data detection problem. Similarly, the effect of $\bsm{\phi}^\mcA_t$ is merely the phase rotation of the noiseless received signal, and so its it is feasible for the RIS to accurately recover the data symbols as $\bsm{\phi}^\mcA_t$ varies in time.

\subsubsection{Estimating $\hd{2},\, \ldots,\, \hd{K}$ and $\bsm{\lambda}^\mcB_2,\, \ldots,\, \bsm{\lambda}^\mcB_K$} 
Since the typical user also transmits data during sub-phase 2b, the received signal at the BS can be re-written in the following form:
\begin{align}
    &\mbf{y}_t^\BS - \Hck{1}\diag{\left(\begin{bmatrix}[0.8]
    1\\ \bsm{\rho}
\end{bmatrix}\right)}\begin{bmatrix}[0.8]
    1\\ \bsm{\phi}_t
\end{bmatrix}s_{1,t} \notag \\
&\qquad = \sqrt{P}\sum_{k=2}^K (\mbf{B}_t \bsm{\upsilon}_k + \mbf{f}_{k,t}^\mcA) s_{k,t} + \mbf{n}_t^\BS\notag\\
    &\qquad =\sqrt{P}((\mbf{S}_{2:K,t}^T\otimes\mbf{B}_t)\bsm{\upsilon} + \mbf{F}_t^\mcA\mbf{S}_{2:K,t}) + \mbf{n}_t^\BS\notag\\
    &\qquad =\sqrt{P}(\mbf{Q}_t\bsm{\upsilon} + \mbf{\tilde{y}}^\BS_t) + \mbf{n}_t^\BS,
    \label{eq:rx_BS_phase2_rewrite}
\end{align}
where $t \in \TII$. Note that $s_{1,t} = 0$ for for $t \in \mca{T}_{\mrm{2a}}$ since user~1 is idle during  sub-phase~2a. Then, we can use a similar technique as in the PD approach for estimating $\bsm{\upsilon}$, but we need to replace $s_{k,t}$ with $\hat{s}_{k,t}$ for $t \in \mca{T}_{2b}$.

\subsection{Overall Training Overhead}
The overall training overhead for the proposed decision-directed approach is $\tau_p = K$ since phase 1 and phase 2 require only 1 and $K-1$ time slots for pilot signalling, respectively. The BS is guaranteed to accurately recover the channel matrices as long as the data symbols are correctly detected by the RIS. Although only the typical user transmits in phase 1 and thus the spectral efficiency will not be as large as if all the users were transmitting, we will show in the numerical results that the proposed DD approach can still result in an increase in the spectral efficiency compared with the PD approach.

\subsection{Comparison with a Passive RIS and DD at the BS}
In the proposed DD method, the BS can accurately recover the cascaded channel matrices when the data symbols are correctly detected by the RIS. Here, we explain why an alternative scenario in which the RIS has no sensing elements and the DD strategy is applied at the BS cannot guarantee accurate CSI estimation. To show this, it is enough to consider the case with only one user, where the received signal at the BS is given as
\begin{equation}
    \mbf{y}^\BS_t = \Hck{1}\bsm{\phi}_ts_t + \mbf{n}_t^\BS. \label{eq:rx_BS_1user}
\end{equation}
To accurately recover $\Hck{1}$, the phase shift vector $\bsm{\phi}_t$ must vary for different time slots $t$ in order to make $\bsm{\Phi}_{\TI} = [\bsm{\phi}_1,\,\ldots,\,\bsm{\phi}_{\tau_1}]$ full-rank. However, if we change $\bsm{\phi}_t$, the effective channel $\mbf{f}_t = \Hck{1}\bsm{\phi}_t$ changes as well. A pilot signal is sent in the first time slot, and the BS can only estimate the effective channel $\mbf{f}_1$. With an estimate of $\mbf{f}_1$, the BS cannot guarantee correct detection of data symbols in subsequent time slots $2,\,\ldots,\,\tau_1$ since the effective channels in these time slots are different from $\mbf{f}_1$ due the change in $\bsm{\phi}_t$. For the effective channel to remain unchanged, the phase shift vector $\bsm{\phi}_t$ must be time-invariant, but in this case the matrix $\bsm{\Phi}_{\TI}$ would be rank-1, which prevents recovery of the cascaded channel $\Hck{1}$. This is unlike the proposed DD channel estimation framework presented above where the data symbols can be accurately recovered by the RIS even when the phase shift vector of the RIS changes in time.

\section{Spectral Efficiency Analysis}
\label{sec_SE}
In this section, we perform a spectral efficiency analysis for both the PD and DD approaches. We consider an RIS-aided system where the BS has one antenna serving one user without a direct channel. We assume $D$-PSK data signalling, i.e., $s \in \mca{S} = \{\exp{\left(j\pi\frac{2\ell+1}{D}\right)}\}$ for $\ell \in \{0,\,\ldots,\,D-1\}$ with the Gray code mapping data bits to data symbols. We also assume that the data symbols are equally likely. Let $a_i = h_ig_i$ be the cascaded channel where $g_i$ and $h_i$ are the channels from the user and the BS to the $i$-th element of the RIS, respectively. It is assumed that $g_i \sim \mca{CN}(0, \sigma_g^2)$ and $h_i \sim \mca{CN}(0, \sigma_h^2)$ are independent of each other. Let $\mbf{a} = [a_1,\,\,\ldots,\,\,a_N]^T$ so that the received signal at the BS can be written as
\begin{equation}
    y^\BS = \sqrt{P}\mbf{a}^H\bsm{\phi}s + n^\BS.
\end{equation}
Let $\mbf{\hat{a}} = \mbf{a} + \bsm{\epsilon}$ be the estimated cascaded channel where $\bsm{\epsilon} = [\epsilon_1,\,\ldots,\,\epsilon_N]^T$ is the channel estimation error.
Given the $\mbf{\hat{a}}$, the RIS coefficients $\bsm{\phi}$ in the data transmission phase are chosen to maximize the effective channel strength, i.e.,
\begin{equation*}
    \maximize_{\{\bsm{\phi}\}}
    \;|\mbf{\hat{a}}^H\bsm{\phi}|^2\;
     \operatorname{subject\ to}
    \; |\phi_i| \leq 1 \; \forall i = 1,\,\ldots,\, N,
\end{equation*}
which has the optimal solution $\phi_i^\star = e^{j\measuredangle(\hat{a}_i)}$. The received signal at the BS in the data transmission phase will then be
\begin{align}
    y^\BS &= \sqrt{P}\sum_{i=1}^N a_i^*e^{j\measuredangle(a_i + \epsilon_i)}s + n^\BS = \sqrt{P}\sum_{i=1}^N z_i s + n^\BS\notag
    \label{eq:ana1}
\end{align}
where $z_i \overset{\Delta}{=} a_i^*e^{j\measuredangle(a_i + \epsilon_i)}$. Thus we have
\begin{align*}
    z_{i,\Re} &\overset{\Delta}{=} \Re\{z_i\} = \frac{|a_{i}|^2 + \Re\{a_i\epsilon_i^*\}}{\sqrt{|a_i|^2 + 2\Re\{a_i\epsilon_i^*\} + |\epsilon_i|^2}}, \\
    z_{i,\Im} &\overset{\Delta}{=} \Im\{z_i\}= \frac{\Im\{a_i\epsilon_i^*\}}{\sqrt{|a_i|^2 + 2\Re\{a_i\epsilon_i^*\} + |\epsilon_i|^2}}.
\end{align*}



\subsection{Pilot-Directed}
For the pilot-directed method, the SE is given as
\begin{equation}
    \mrm{SE}_{\mtt{PD}} = \frac{\tauc - \taup}{\tauc}\left(1-\mrm{BER}_{\mtt{PD}}\right)\log_2(D),
\end{equation}
where $\tauc$ and $\taup$ are the lengths of the coherence block and the pilot sequence, respectively. In the PD approach, the first $N$ time slots are used for channel estimation (i.e., $\taup = N$), and we assume without loss of generality that the pilot signal is an all-ones vector. Assuming a DFT matrix of size $N$ is used to configure the RIS phase shifts during the channel estimation phase, we have that $\epsilon_i \sim \mca{CN}\big(0, \frac{N_0^\BS}{PN}\big)$. The CSI errors $\{\epsilon_i\}$ are also i.i.d.. and uncorrelated with $a_i$. We will compute the PD bit error rate $\mrm{BER}_{\mtt{PD}}$, which requires the distribution of the effective channel $f = \sum_{i=1}^N z_i$.

We first obtain the following approximate means
\begin{align}
    \mu_{z_{i,\Re}} \overset{\Delta}{=} \E{z_{i,\Re}} &\approx \frac{\E{|a_{i}|^2} + \Re\{\E{a_i\epsilon_i^*}\}}{\sqrt{\E{|a_i|^2} + 2\Re\{\E{a_i\epsilon_i^*}\} + \E{|\epsilon_i|^2}}} \notag \\
    &= \frac{\sigma_a^2}{\sqrt{\sigma_a^2+\sigma_\epsilon^2}}\\
    \mu_{z_{i,\Im}} \overset{\Delta}{=} \E{z_{i,\Im}} &\approx \frac{\Im\{\E{a_i\epsilon_i^*}\}}{\sqrt{\E{|a_i|^2} + 2\Re\{\E{a_i\epsilon_i^*}\} + \E{|\epsilon_i|^2}}} = 0
\end{align}
and variances
\begin{align}
    \sigma_{z_{i,\Re}}^2 &\overset{\Delta}{=} \mrm{Var}[z_{i,\Re}] = \mbb{E}[z_{i,\Re}^2] - \mbb{E}[z_{i,\Re}]^2 \approx \frac{7\sigma_a^4 + \sigma_a^2\sigma_\epsilon^2}{2(\sigma_a^2 + \sigma_\epsilon^2)}\\
    \sigma_{z_{i,\Im}}^2 &\overset{\Delta}{=} \mrm{Var}[z_{i,\Im}] = \mbb{E}[z_{i,\Im}^2] - \mbb{E}[z_{i,\Im}]^2 \approx \frac{\sigma_a^2\sigma_\epsilon^2}{2(\sigma_a^2 + \sigma_\epsilon^2)}
\end{align}
where $\sigma_a^2 \overset{\Delta}{=} \sigma_g^2\sigma_h^2$ and we have used the following results:
\begin{align}
    \E{a_{i,\Re}^2} &= \E{(h_{i,\Re} g_{i,\Re} - h_{i,\Im} g_{i,\Im})^2} = \frac{1}{2}\sigma_h^2\sigma_g^2 = \frac{1}{2}\sigma_a^2,\notag\\
    \E{a_{i,\Re}^4} &= \E{(h_{i,\Re} g_{i,\Re} - h_{i,\Im} g_{i,\Im})^4} = \frac{3}{2}\sigma_h^4\sigma_g^4 = \frac{3}{2}\sigma_a^4.\notag
\end{align}
It can be seen that the means and variances $\mu_{z_{i,\Re}}$, $\mu_{z_{i,\Im}}$, $\sigma_{z_{i,\Re}}^2$, and $\sigma_{z_{i,\Im}}^2$ above are the same for different indices $i$. Therefore, for convenience in the rest of the PD-SE analysis, we drop the index $i$ for these values.

Since the $\{a_i\}$ and $\{\epsilon_i\}$ are i.i.d., then the $\{z_i\}$ are also i.i.d.. Using the central-limit theorem, for large $N$ we have $f_\Re = \sum_{i=1}^N z_{i,\Re} \sim \mca{CN}(N\mu_{r_\Re}, N\sigma_{z_\Re}^2)$ and $f_\Im = \sum_{i=1}^N z_{i,\Im} \sim \mca{CN}(N\mu_{z_\Im}, N\sigma_{z_\Im}^2)$. Note that we have $\Cov{f_\Re, f_\Im} = 0$ since the $\{z_i\}$ are i.i.d, and $\Cov{z_{i,\Re}, z_{i,\Im}} = 0$.

Let $\Tilde{y} = ys^* = \sqrt{P}f + \tilde{n}$ be the rotated received signal, and define $r_{\tilde{y}} = \sqrt{\tilde{y}_{\Re}^2 + \tilde{y}_{\Im}^2}$ and $\theta_{\tilde{y}} = \measuredangle (\tilde{y}) = \arctan(\tilde{y}_{\Im}/\tilde{y}_{\Re}) $. Then the joint pdf of $r_{\tilde{y}}$ and $\theta_{\tilde{y}}$ is given as
\begin{align}
    &p(r_{\tilde{y}},\theta_{\tilde{y}}) = \frac{r_{\tilde{y}}}{2\pi\sqrt{(PN\sigma^2_{z_\Re} + N_0/2)(PN\sigma^2_{z_\Im} + N_0/2)}} \times \notag \\
    &\exp \left\{-\frac{1}{2} \left[\frac{(r_{\tilde{y}}\cos \theta_{\tilde{y}}-\sqrt{P}N\mu_{z_\Re})^2}{PN\sigma_{z_\Re}^2 + N_0/2} + \frac{(r_{\tilde{y}}\sin \theta_{\tilde{y}})^2}{PN\sigma_{z_\Im}^2 + N_0/2}\right] \right\}.
    \label{eq_PD_pdf}
\end{align}
Since the data symbols are equally likely, to compute the BER, we can assume that any one of the data symbols was transmitted, and we choose $\mca{S}(0)$. The probability that the detected symbol is $\mca{S}({\ell})$ given $\mca{S}(0)$ was transmitted is
\begin{align}
    p_{\ell}^\PD &\overset{\Delta}{=} \mbb{P}[\hat{s} = \mca{S}({\ell})\mid s = \mca{S}(0)] \notag \\
    &= \int_{0}^{\infty}\int_{\frac{(2{\ell}-1)\pi}{D}}^{\frac{(2{\ell}+1)\pi}{D}} p(r_{\tilde{y}},\theta_{\tilde{y}}) d\theta_{\tilde{y}} dr_{\tilde{y}}.
    \label{eq_PD_p}
\end{align}
Thus, the BER is given as
\begin{equation}
    \mrm{BER}_{\mtt{PD}} = \frac{1}{\log_2(D)}\sum_{\ell=1}^{D}p_{\ell}^\PD e_{\mrm{bit}}(0,\ell),
    \label{eq_BER_PD}
\end{equation}
where $e_{\mrm{bit}}(0,\ell)$ is the number of bit differences between symbols $\mca{S}(0)$ and $\mca{S}(\ell)$.

\begin{theorem}
At high SNR, $\mrm{BER}_{\mtt{PD}}$ can be approximated as
    \begin{align}
    &\mrm{BER}_{\mtt{PD}} \approx \notag\\
    &\quad\frac{2}{\log_2(D)}Q\left(\frac{\pi\sqrt{P}N \sigma_a \tan \theta}{4\sqrt{\left(PN\left(1 -\frac{\pi^2}{16}\right)\sigma_a^2+\frac{N_0}{2}\right)\tan^2\theta +\frac{N_0}{2}}}\right).
    \label{eq_BER_PD_high_SNR}
\end{align}
\label{theorem_1}
\end{theorem}
\begin{proof}
See Appendix~\ref{appendia_1}.
\end{proof}

\subsection{Decision-Directed}
The SE of the decision-directed approach is given as
\begin{equation}
    \mrm{SE}_{\mtt{DD}} = \left(\frac{\tau_{\mrm{d},1}(1-\mrm{BER}_{\mtt{DD}1}) + \tau_{\mrm{d},2}(1-\mrm{BER}_{\mtt{DD}2})}{\tauc}\right)\log_2(D)
\end{equation}
where $\tau_{\mrm{d},1}$ and $\mrm{BER}_{\mtt{DD}1}$ are the data transmission length and the BER in the channel estimation stage. Similarly, $\tau_{\mrm{d},2}$ and $\mrm{BER}_{\mtt{DD}2}$ are the data transmission length and the BER in the data transmission stage. Thus, for the DD spectral analysis analysis, we need to compute $\mrm{BER}_{\mtt{DD}1}$ and $\mrm{BER}_{\mtt{DD}2}$ to obtain $\mrm{SE}_{\mtt{DD}}$. While $\mrm{BER}_{\mtt{DD}1}$ is simple to obtain and can be computed exactly, obtaining an exact value $\mrm{BER}_{\mtt{DD}2}$ is much more challenging and thus we provide an accurate approximation.

To simplify the analysis, we assume the RIS has only one active receiver element, which we take to be element $N$. We further assume that this element completely absorbs the incoming signal power during the channel estimation stage and is then turned off in the data transmission stage. In the first time slot, a pilot signal $s_1 = 1$ is transmitted to generate the following received signal at the RIS: $y_{1}^\RIS = \sqrt{P}g_N + n_1^\RIS$, and so an estimate of $g_N$ can be obtained as $\hat{g}_N = y_{1}^\RIS/\sqrt{P}= g_N + \frac{1}{\sqrt{P}}n_1^\RIS$. 
From the second time slot to the $(N-1)$-th time slot, data symbols are transmitted and the received signal at the RIS is $y_{t}^\RIS = \sqrt{P}g_N s_{t} + n_t^\RIS,\;t = 2,\,\ldots,\,N-1$. An equalizer based on $\hat{g}_N$ is used by the RIS to detect the symbols as
\begin{align}
    \tilde{s}_t = \frac{y_t^\RIS}{\sqrt{P}\hat{g}_N} = \frac{\sqrt{P}g_N s_{t} + n_t^\RIS}{\sqrt{P}g_N + n_1^\RIS},\;t = 2,\,\ldots,\,N-1
    \label{eq_RIS_DD}
\end{align}

To evaluate the BER of~\eqref{eq_RIS_DD}, we consider the rotated signal
\begin{align}
    \Bar{s}_t = \tilde{s}_ts_t^* = \frac{\sqrt{P}g_N + \tilde{n}_t^\RIS}{\sqrt{P}g_N + n_1^\RIS}.
\end{align}
Since the error rates are the same for different time indices $t$, we drop the subscript $t$ for convenience. 
 Let $r_{\bar{s}} = \sqrt{\bar{s}_{\Re}^2 + \bar{s}_{\Im}^2}$ and $\theta_{\bar{s}} = \arctan(\bar{s}_{\Im}/\bar{s}_{\Re}) $, whose joint pdf can be obtained as follows~\cite{Baxley2010Complex}:
\begin{align}
    p(r_{\bar{s}},\theta_{\bar{s}}) = \frac{r_{\bar{s}}(\sigma_{y^\RIS}^4 - P^2\sigma_g^4)}{\pi\sigma_{y^\RIS}^8}\bigg(\frac{1 + r_{\bar{s}}^2}{\sigma_{y^\RIS}^2} -  \frac{2P\sigma_g^2r_{\bar{s}}\cos\theta_{\bar{s}}}{\sigma_{y^\RIS}^4}\bigg)^{-2},\notag
\end{align}
where $\sigma_{y^\RIS}^2 = P\sigma_g^2 + N_0^\RIS$. The probability that the detected symbol is $\mca{S}({\ell})$ given that $\mca{S}(0)$ was transmitted is
\begin{align}
    p_{\ell}^{\DDone} = \int_{0}^{\infty}\int_{\frac{(2{\ell}-1)\pi}{D}}^{\frac{(2{\ell}+1)\pi}{D}} p(r_{\bar{s}},\theta_{\bar{s}}) d\theta_{\bar{s}} dr_{\bar{s}}.
    \label{eq_p_ell_DD1}
\end{align}
Similar to~\eqref{eq_BER_PD}, we can compute $\mrm{BER}_{\mrm{DD}1}$ using $p_{\ell}^{\DDone}$ in~\eqref{eq_p_ell_DD1}.

Next we compute $\mrm{BER}_{\mrm{DD}2}$. Let $\mbf{\hat{s}} = [1,\,\hat{s}_2,\,\ldots,\,\hat{s}_{N-1}]^T$ be the detected data symbol vector forwarded to the BS by the RIS, where $\hat{s}_t = \argmin_{s\in\mca{S}} |s - \tilde{s}_t|^2$. Let $\bsm{\Phi} = \mbf{V}_{N-1}$ be the phase shift matrix of the RIS during the channel estimation stage. The BS uses $\mbf{\hat{s}}$ to obtain an estimate of the cascaded channel as follows:
\begin{align}
    \hat{a}_i^* &= \frac{1}{N-1}\mbf{a}^H\bsm{\Phi} \diag(\mbf{s}) \diag(\mbf{\hat{s}})^{-1}\bsm{\Phi}_{i,:}^H  + \tilde{n}_i^*\\
    &= a_i^* + \underbrace{\frac{1}{1 - N}\mbf{a}^H\bsm{\Phi} \diag{(\bsm{\xi})}\bsm{\Phi}_{i,:}^H + \tilde{n}_i^*}_{\epsilon_i^*}
\end{align}
where $\bsm{\xi} = [\xi_1,\, \ldots,\, \xi_{N-1}]^T$ with $\xi_t = 1 - s_t\hat{s}_t^*$. Note that $\xi_1 = 0$ since $s_1=1$ is a known pilot signal. The noise terms $\tilde{n}_i^* = \frac{1}{\sqrt{P}(N-1)}\mbf{n}^T\diag{(\mbf{\hat{s}})}^{-1}\bsm{\Phi}_{i,:}^H$ are i.i.d. for different indices $i$ and the channel estimation error term is $\epsilon_i^* = \delta_i^* + \tilde{n}_i^*$ where
\begin{equation}
    \delta_i^* = \frac{1}{1-N}\mbf{a}^H\bsm{\Phi} \diag {\left(\bsm{\xi}\right)}\bsm{\Phi}_{i,:}^H.
    \label{eq_chan_est_error}
\end{equation}
From~\eqref{eq_chan_est_error}, it can be seen that the channel estimation error depends on the data symbols detected at the RIS.

Our DD-SE analysis will rely on the results in Lemma~\ref{lemma_1} presented next.
\begin{lemma}
Let $\mu_\xi \overset{\Delta}{=} \E{\xi_t}$, $\mu_{|\xi|^2} \overset{\Delta}{=} \E{\left|\xi_t\right|^2}$, and $\mu_{\xi^2} \overset{\Delta}{=} \E{\left(\xi_t\right)^2}$, then we have
\begin{align}
    \mu_\xi &= 1 - p_{0}^{\DDone} + p_{\frac{D}{2}}^{\DDone} - 2\sum_{d=1}^{\frac{D}{2}-1}p_{d}^{\DDone}\cos\left(\frac{2\pi d}{D}\right), \label{eq:mu_xi}\\
    \mu_{\xi^2} &= 1 - p_{0}^{\DDone} + 3p_{\frac{D}{2}}^{\DDone} \, + \notag \\
    &\qquad 2\sum_{d=1}^{\frac{D}{2} - 1}p_{d}^{\DDone} \left[\cos\left(\frac{4\pi d}{D}\right) - 2\cos\left(\frac{2\pi d}{D}\right)\right], \label{eq:mu_xi2}\\
    \mu_{|\xi|^2} &= 4p_{\frac{D}{2}}^{\DDone} + 4\sum_{{d}=1}^{\frac{D}{2}-1} p_{{d}}^{\DDone} \left[1-\cos \left(\frac{2\pi{d}}{D}\right)\right] .\label{eq:mu_absxi2}
\end{align}
\label{lemma_1}
\end{lemma}

\begin{proof}
See Appendix~\ref{appendia_2}.
\end{proof}

To obtain the $\mrm{BER}_{\mtt{DD}2}$, we will compute the means and variances of the real and imaginary parts of the effective channel $f = \sum_{i=1}^{N-1} z_i$ during the data transmission stage, which are $\mu_{f_\Re} \overset{\Delta}{=} \E{f_\Re}$ and $\mu_{f_\Im} \overset{\Delta}{=} \E{{f_\Im}}$, $\sigma_{f_\Re}^2 \overset{\Delta}{=} \Var{{f_\Re}}$ and $\sigma_{f_\Im}^2 \overset{\Delta}{=} \Var{{f_\Im}}$.

\subsubsection{Computing the means $\mu_{f_\Re}$ and $\mu_{f_\Im}$}
Since $\mu_{f_\Re} = \sum_{i=1}^{N-1} \mu_{z_{i,\Re}}$ and $\mu_{f_\Im} = \sum_{i=1}^{N-1} \mu_{z_{i,\Im}}$, we need to compute $\mu_{z_{i,\Re}}$ and $\mu_{z_{i,\Im}}$, which are approximated as follows:
\begin{align}
    \mu_{z_{i,\Re}} \approx \frac{\E{|a_i|^2} + \Re\{\E{a_i\epsilon_i^*}\}}{\sqrt{\E{|a_i|^2} + 2\Re\{\E{a_i\epsilon_i^*}\} + \E{|\epsilon_i|^2}}}, \\
    \mu_{z_{i,\Im}} \approx \frac{\Im\{\E{a_i\epsilon_i^*}\}}{\sqrt{\E{|a_i|^2} + 2\Re\{\E{a_i\epsilon_i^*}\} + \E{|\epsilon_i|^2}}}.
\end{align}
Since
\begin{align}
    \E{a_i\epsilon_i^*} = \E{a_i\delta_i^*} 
    &= \frac{1}{1-N}\E{\sum_{n=1}^{N-1}a_ia_n^*\bsm{\Phi}_{n,:}\diag{\left(\bsm{\xi}\right)}\bsm{\Phi}_{i,:}^H} \notag \\
    &= \frac{1}{1-N}\E{|a_i|^2\bsm{\Phi}_{i,:}\diag{\left(\bsm{\xi}\right)}\bsm{\Phi}_{i,:}^H} \notag \\
    &= \frac{\sigma_a^2}{1-N}\sum_{t=2}^{N-1}\E{\xi_t} = \frac{N-2}{1-N}\sigma_a^2\mu_{\xi}, \notag
\end{align}
we have $\E{\Re\{a_i\epsilon_i^*\}} = \frac{N-2}{1-N}\sigma_a^2\mu_{\xi}$ and $\E{\Im\{a_i\epsilon_i^*\}} = 0$. Thus, the mean of $z_{i,\Re}$ and $z_{i,\Im}$ can be obtained as
\begin{align}
    \mu_{z_{i,\Re}} \approx \frac{\sigma_a^2 + \frac{N-2}{1-N}\sigma_a^2\mu_{\xi}}{\sqrt{\sigma_a^2+2\frac{N-2}{1-N}\sigma_a^2\mu_{\xi}+\sigma_\epsilon^2}}  \;\text{and}\;
    \mu_{z_{i,\Im}} \approx 0, \label{eq:mu_rR-mu_rI}
\end{align}
where $\sigma_a^2 \overset{\Delta}{=} \E{|a|^2}$ and 
$\sigma_\epsilon^2 \overset{\Delta}{=} \E{|\epsilon_i|^2} = \sigma_{\delta_i}^2 + \sigma_{\tilde{n}_i}^2$.
The variance of $\delta_i$ and $\tilde{n}_i$ ($\sigma_{\delta_i}^2$ and $\sigma_{\tilde{n}_i}^2$) are given as follows:
\begin{align}
    \sigma_{\delta_i}^2 &= \frac{\sigma_a^2}{(N-1)}\tr\left(\E{ \diag {\left(\bsm{\xi}\right)}\bsm{\Phi}_{i,:}^H\bsm{\Phi}_{i,:}\diag {\left(\bsm{\xi}^*\right)}}\right) \notag \\
    &= \frac{\sigma_a^2}{(N-1)}\sum_{t=2}^{N-1}\E{|\xi_t|^2} = \frac{N-2}{(N-1)}\sigma_a^2\mu_{|\xi|^2}
\end{align}
and $\sigma_{\tilde{n}_i}^2 = \frac{N_0^\BS}{P(N-1)}$. Hence, the variance of $\epsilon$ is
\begin{equation}
    \sigma_\epsilon^2 = \frac{P(N-2)\sigma_a^2\mu_{|\xi|^2} + N_0^\BS}{P(N-1)}.
\end{equation}

\subsubsection{Compute the variances $\sigma_{f_\Re}^2$ and $\sigma_{f_\Im}^2$} The variance of $f_\Re$ and $f_\Im$ are 
\begin{align}
    \sigma_{f_\Re}^2 &= \sum_{i=1}^{N-1} \Var{z_{i,\Re}} + 2 \sum_{i<t} \Cov{z_{i,\Re}, z_{t,\Re}},\\
    \sigma_{f_\Im}^2 &= \sum_{i=1}^{N-1} \Var{z_{i,\Im}} + 2 \sum_{i<t} \Cov{z_{i,\Im}, z_{t,\Im}}.
\end{align}
Since $\mrm{Var}[z_{i,\Re}] = \mbb{E}[z_{i,\Re}^2] - \mbb{E}[z_{i,\Re}]^2$ and $\mrm{Var}[z_{i,\Im}] = \mbb{E}[z_{i,\Im}^2] - \mbb{E}[z_{i,\Im}]^2$, and since $\mbb{E}[z_{i,\Re}]$ and $\mbb{E}[z_{i,\Im}]$ are already given in~\eqref{eq:mu_rR-mu_rI}, we now need to obtain $\mbb{E}[z_{i,\Re}^2]$ and $\mbb{E}[z_{i,\Im}^2]$, which can be approximated as follows:
\begin{align}
    \E{z_{i,\Re}^2} &\approx \frac{\E{|a_i|^4} +  2\E{|a_i|^2\Re\{a_i\epsilon_i^*\}} + \E{\Re\{a_i\epsilon_i^*\}^2}}{{\E{|a_i|^2} + 2\Re\{\E{a_i\epsilon_i^*}\} + \E{|\epsilon_i|^2}}}, \label{eq:ErR2}\\
    \E{z_{i,\Im}^2} &\approx \frac{\E{\Im\{a_i\epsilon_i^*\}^2}}{{\E{|a_i|^2} + 2\Re\{\E{a_i\epsilon_i^*}\} + \E{|\epsilon_i|^2}}}.\label{eq:ErI2}
\end{align}
The first term in the numerator of~\eqref{eq:ErR2} is $\E{|a_i|^4} = 4\sigma_a^4$. The second term in the numerator of~\eqref{eq:ErR2} can be obtained by computing $\E{|a_i|^2a_i\epsilon_i^*}$ since $\E{|a_i|^2\Re\{a_i\epsilon_i^*\}} = \Re\{\E{|a_i|^2a_i\epsilon_i^*}\}$. We have
\begin{align}
    \E{|a_i|^2a_i\epsilon_i^*} &= \E{|a_i|^2a_i\delta_i^*} \notag\\
    &= \frac{1}{1-N}\E{\sum_{n=1}^{N-1}|a_i|^2a_ia_n^*\bsm{\Phi}_{n,:}\diag{\left(\bsm{\xi}\right)}\bsm{\Phi}_{i,:}^H}\notag\\
    &= \frac{1}{1-N}\E{|a_i|^4\bsm{\Phi}_{i,:}\diag{\left(\bsm{\xi}\right)}\bsm{\Phi}_{i,:}^H}\notag\\
    &= \frac{1}{1-N}4\sigma_a^4\sum_{t=2}^{N-1}\E{\xi_t} = \frac{N-2}{1-N}4\sigma_a^4\mu_\xi, \label{eq:E_absa2_a_delta}
\end{align}
which is a real number, and so $\E{|a_i|^2\Re\{a_i\epsilon_i^*\}} = \frac{N-2}{1-N}4\sigma_a^4\mu_\xi$.

We exploit the following relations:
\begin{align}
    |a_i\epsilon_i^*|^2 &= \Re\{a_i\epsilon_i^*\}^2 + \Im
    \{a_i\epsilon_i^*\}^2,\\
    \Re\{(a_i\epsilon_i^*)^2\} &= \Re\{a_i\epsilon_i^*\}^2 - \Im\{a_i\epsilon_i^*\}^2,
\end{align}
to obtain $\E{\Re\{a_i\epsilon_i^*\}^2}$ (the last term in the numerator of~\eqref{eq:ErR2}) and $\E{\Im\{a_i\epsilon_i^*\}^2}$ (the numerator term of~\eqref{eq:ErI2}) as follows:
\begin{align}
    \E{\Re\{a_i\epsilon_i^*\}^2} &= \frac{\E{|a_i\epsilon_i^*|^2} + \Re\big\{\E{(a_i\epsilon_i^*)^2}\big\}}{2},\\
    \E{\Im\{a_i\epsilon_i^*\}^2} &= \frac{\E{|a_i\epsilon_i^*|^2} - \Re\big\{\E{(a_i\epsilon_i^*)^2}\big\}}{2}.
\end{align}
We need to find $\E{|a_i\epsilon_i^*|^2}$ and $\E{(a_i\epsilon_i^*)^2}$ to retrieve $\E{\Re\{a_i\epsilon_i^*\}^2}$ and $\E{\Im\{a_i\epsilon_i^*\}^2}$. The term $\E{|a_i\epsilon_i^*|^2}$ is given by 
\begin{align}
    \E{|a_i\epsilon_i^*|^2} &= \E{|a_i\delta_i^*|^2} + \E{|a_i\tilde{n}_i^*|^2} \notag \\
    &= \E{|a_i\delta_i^*|^2} + \frac{\sigma_a^2N_0^\BS}{P(N-1)}
\end{align}
where
\begin{align}
    &\E{|a_i\delta_i^*|^2} \notag \\
    &= \frac{\E{\left|a_i\mbf{a}^H\bsm{\Phi} \diag {\left(\bsm{\xi}\right)}\bsm{\Phi}_{i,:}^H\right|^2}}{(N-1)^2}\notag\\
    &= \frac{\E{\tr{\left\{|a_i|^2\bsm{\Phi}^H\mbf{a}\mbf{a}^H\bsm{\Phi}\diag {\left(\bsm{\Phi}_{i,:}^H\right)}\bsm{\xi}\bsm{\xi}^H\diag {\left(\bsm{\Phi}_{i,:}\right)}\right\}}}}{(N-1)^2} \notag\\
    &= \frac{\tr{\left\{\bsm{\Phi}^H\diag{(\bsm{\alpha}_i)}\bsm{\Phi}\diag {\left(\bsm{\Phi}_{i,:}^H\right)}\mbf{R}_{\bsm{\xi}\bsm{\xi}^H}\diag {\left(\bsm{\Phi}_{i,:}\right)}\right\}}}{(N-1)^2} \notag\\
    &=\frac{N-2}{(N-1)^2}\sigma_a^4\left((N+2)\mu_{|\xi|^2} + 3(N-3)\mu_{\xi}^2\right)
\end{align}
and $\mbf{R}_{\bsm{\xi}\bsm{\xi}^H} = \E{\bsm{\xi}\bsm{\xi}^H}$. Here, the diagonal elements of $\mbf{R}_{\bsm{\xi}\bsm{\xi}^H}$ are $\mu_{|\xi|^2}$, the off-diagonal elements are approximately
$\mu_{\xi}$, and $\bsm{\alpha}_i$ is a vector whose $i$-element is $4\sigma_a^4$ and whose other elements are $\sigma_a^4$.

The term $\E{(a_i\epsilon_i^*)^2}$ is given by
\begin{align}
    &\E{(a_i\epsilon_i^*)^2} = \E{(a_i\delta_i^*)^2}\notag\\
    &= \frac{\E{\left(a_i\mbf{a}^H\bsm{\Phi} \diag {\left(\bsm{\xi}\right)}\bsm{\Phi}_{i,:}^H\right)^2}}{(N-1)^2}\notag\\
    &=\frac{\E{\tr{\left\{a_i^2\diag {\left(\bsm{\Phi}_{i,:}^*\right)}\bsm{\Phi}^T\mbf{a}^*\mbf{a}^H\bsm{\Phi}\diag {\left(\bsm{\Phi}_{i,:}^H\right)}\bsm{\xi}\bsm{\xi}^T\right\}}}}{(N-1)^2}\notag\\
    &=\frac{\tr{\left\{4\sigma_a^4\diag {\left(\bsm{\Phi}_{i,:}^*\right)}\bsm{\Phi}^T\diag{(\mbf{e}_i)}\bsm{\Phi}\diag {\left(\bsm{\Phi}_{i,:}^H\right)}\mbf{R}_{\bsm{\xi}\bsm{\xi}^T}\right\}}}{(N-1)^2}\notag\\
    &=\frac{N-2}{(N-1)^2}4\sigma_a^4\big(\mu_{\xi^2} + (N-3)\mu_{\xi}^2\big)
\end{align}
where $\mbf{R}_{\bsm{\xi}\bsm{\xi}^T} = \E{\bsm{\xi}\bsm{\xi}^T}$. Here, the diagonal elements of $\mbf{R}_{\bsm{\xi}\bsm{\xi}^T}$ are $\mu_{\xi^2}$, the off-diagonal elements are approximately
$\mu_{\xi}^2$, and $\bsm{e}_i$ is a one-hot vector whose $i$-element is 1 and whose  other elements are zero. 

We now need to compute the covariances $\Cov{z_{i,\Re}, z_{t,\Re}}$ and $ \Cov{z_{i,\Im}, z_{t,\Im}}$, which are given by
\begin{align*}
    \Cov{z_{i,\Re}, z_{t,\Re}} &= \E{z_{i,\Re}z_{t,\Re}} - \E{z_{i,\Re}}\E{z_{t,\Re}}\\
    \Cov{z_{i,\Im}, z_{t,\Im}} &= \E{z_{i,\Im}z_{t,\Im}} - \E{z_{i,\Im}}\E{z_{t,\Im}}.
\end{align*}
Since $\E{z_{i,\Re}}$ and $\E{z_{i,\Im}}$ have been computed earlier, we need to find $\E{z_{i,\Re}z_{t,\Re}}$ and $\E{z_{i,\Im}z_{t,\Im}}$ which can be approximated as follows:
\begin{align}
    \E{z_{i,\Re}z_{t,\Re}} \approx  \frac{\E{(|a_i|^2 + \Re\{a_i\epsilon_i^*\})(|a_t|^2 + \Re\{a_t\epsilon_t^*\})}}{\sqrt{\E{\varkappa_{i,t}}}},
    \label{eq_E_zi_zt_Re}
\end{align}
and
\begin{align}
    \E{z_{i,\Im}z_{t,\Im}} \approx \frac{\E{\Im\{a_i\epsilon_i^*\}\Im\{a_t\epsilon_t^*\}}}{\sqrt{\E{\varkappa_{i,t}}}},
    \label{eq_E_zi_zt_Im}
\end{align}
where $\varkappa_{i,t}$ is given in~\eqref{eq_expanded_denom}. Since the expansion of $\varkappa_{i,t}$ also includes the numerator terms in~\eqref{eq_E_zi_zt_Re} and~\eqref{eq_E_zi_zt_Im}, to obtain $\E{z_{i,\Re}z_{t,\Re}}$ and $\E{z_{i,\Im}z_{t,\Im}}$, it suffices to compute the expectation of all the terms in~\eqref{eq_expanded_denom}.
\begin{figure*}[t!]
    \begin{align}
        \varkappa_{i,t} & = |a_i|^2|a_t|^2 + 2|a_i|^2\Re\{a_t\epsilon_t^*\} + |a_i|^2|\epsilon_t|^2 + 2|a_t|^2\Re\{a_i\epsilon_i^*\}  + 4\Re\{a_i\epsilon_i^*\}\Re\{a_t\epsilon_t^*\} \; + \;\notag\\
        &\hspace{2cm} 2\Re\{a_i\epsilon_i^*\}|\epsilon_t|^2 + |\epsilon_i|^2|a_t|^2 + 2|\epsilon_i|^2\Re\{a_t\epsilon_t^*\} + |\epsilon_i|^2|\epsilon_t|^2.
        \label{eq_expanded_denom}
    \end{align}
\end{figure*}

First, we have $\E{|a_i|^2|a_t|^2} = \sigma_a^4$. Using the same approach as in~\eqref{eq:E_absa2_a_delta}, we obtain
\begin{align}
    \E{|a_i|^2a_t\epsilon_t^*} &= \E{|a_t|^2a_i\epsilon_i^*} = \frac{N-2}{1-N}\sigma_a^4\mu_\xi.
\end{align}
The two terms $\E{|a_i\epsilon_t|^2}$ and $\E{|a_t\epsilon_i|^2}$ are also equal and given by
\begin{align}
    \E{|a_i\epsilon_t|^2} &= \E{|a_i\delta_t|^2} + \E{|a_i\tilde{n}_t|^2} \notag \\
    &= \E{|a_i\delta_t|^2} + \frac{\sigma_a^2N_0^\BS}{P(N-1)}
\end{align}
where
\begin{align*}
    &\E{|a_i\delta_t|^2} \notag \\
    &= \frac{\E{\left|a_i\mbf{a}^H\bsm{\Phi} \diag {\left(\bsm{\xi}\right)}\bsm{\Phi}_{t,:}^H\right|^2}}{(N-1)^2}\notag\\
    &= \frac{\E{\tr{\left\{|a_i|^2\bsm{\Phi}^H\mbf{a}\mbf{a}^H\bsm{\Phi}\diag {\left(\bsm{\Phi}_{t,:}^H\right)}\bsm{\xi}\bsm{\xi}^H\diag {\left(\bsm{\Phi}_{t,:}\right)}\right\}}}}{(N-1)^2}\notag\\
    &= \frac{\tr{\left\{\bsm{\Phi}^H\diag{(\bsm{\alpha}_i)}\bsm{\Phi}\diag {\left(\bsm{\Phi}_{t,:}^H\right)}\mbf{R}_{\bsm{\xi}\bsm{\xi}^H}\diag {\left(\bsm{\Phi}_{t,:}\right)}\right\}}}{(N-1)^2}.
\end{align*}

The terms $\E{a_i\epsilon_i^*a_t\epsilon_t^*}$ and $\E{a_i\epsilon_i^*a_t^*\epsilon_t}$ are obtained as follows:
\begin{align}
   &\E{a_i\epsilon_i^*a_t\epsilon_t^*} = \E{a_ia_t\delta_i^*\delta_t^*} \notag\\
    &= \frac{\E{a_ia_t\mbf{a}^H\bsm{\Phi} \diag {\left(\bsm{\xi}\right)}\bsm{\Phi}_{i,:}^H\mbf{a}^H\bsm{\Phi} \diag {\left(\bsm{\xi}\right)}\bsm{\Phi}_{t,:}^H}}{(N-1)^2}\notag\\
    &= \frac{\E{\tr{\left\{\bsm{\Phi}^Ta_ia_t\mbf{a}^*\mbf{a}^H\bsm{\Phi} \diag {\left(\bsm{\xi}\right)}\bsm{\Phi}_{i,:}^H\bsm{\Phi}_{t,:}^*\diag {\left(\bsm{\xi}\right)}\right\}}}}{(N-1)^2}\notag\\
    &= \frac{\tr{\left\{\bsm{\Phi}^T \bsm{\Sigma} \bsm{\Phi} \diag {\left(\bsm{\Phi}_{i,:}^H\right)}\mbf{R}_{\bsm{\xi}\bsm{\xi}^T}\diag {\left(\bsm{\Phi}_{t,:}^*\right)}\right\}}}{(N-1)^2}
\end{align}
where $\bsm{\Sigma}$ is a matrix with $\bsm{\Sigma}_{i,t} = \bsm{\Sigma}_{t,i} = \sigma_a^4$, and zeroes elsewhere, and
\begin{align}
   &\E{a_i\epsilon_i^*a_t^*\epsilon_t} = \E{a_ia_t^*\delta_i^*\delta_t} \notag\\
    &= \frac{\E{a_ia_t^*\mbf{a}^H\bsm{\Phi} \diag {\left(\bsm{\xi}\right)}\bsm{\Phi}_{i,:}^H\mbf{a}^T\bsm{\Phi}^* \diag {\left(\bsm{\xi}^*\right)}\bsm{\Phi}_{t,:}^T}}{(N-1)^2}\notag\\
    &= \frac{\E{\tr{\left\{\bsm{\Phi}^Ha_ia_t^*\mbf{a}\mbf{a}^H\bsm{\Phi} \diag {\left(\bsm{\xi}\right)}\bsm{\Phi}_{i,:}^H\bsm{\Phi}_{t,:}\diag {\left(\bsm{\xi}^*\right)}\right\}}}}{(N-1)^2}\notag\\
    &= \frac{\tr{\left\{\bsm{\Phi}^H\bsm{\Omega}\bsm{\Phi} \diag {\left(\bsm{\Phi}_{i,:}^H\right)}\mbf{R}_{\bsm{\xi}\bsm{\xi}^H}\diag {\left(\bsm{\Phi}_{t,:}\right)}\right\}}}{(N-1)^2}
\end{align}
where $\bsm{\Omega}$ is a matrix with $\bsm{\Omega}_{t,i} = \sigma_a^4$, and zeroes elsewhere. 

The two terms $\E{a_i\epsilon_i^*|\epsilon_t|^2}$ and $\E{a_i\epsilon_i^*|\epsilon_t|^2}$ are also equal and given by
\begin{align}
    \E{a_i\epsilon_i^*|\epsilon_t|^2} &= \E{a_i\delta_i^*|\delta_t|^2} + \E{a_i\delta_i^*|\tilde{n}_t|^2}\notag\\
    &\approx \E{a_i\delta_i^*|\tilde{n}_t|^2} = - \frac{(N-2)N_0^\BS}{P(N-1)^2}\sigma_a^2\mu_{\xi}.
    \label{eq_88}
\end{align}
Finally, the term $\E{|\epsilon_i|^2|\epsilon_t|^2}$ is approximated as
\begin{align}
    \E{|\epsilon_i|^2|\epsilon_t|^2} &= \E{|\delta_i|^2|\delta_t|^2} + 2\E{|\delta_i|^2|\tilde{n}_t|^2} + \E{|\tilde{n}_i|^2|\tilde{n}_t|^2}\notag\\
    &\approx 2\E{|\delta_i|^2|\tilde{n}_t|^2} + \E{|\tilde{n}_i|^2|\tilde{n}_t|^2}\notag\\
    &= \frac{2(N-2)N_0^\BS}{P(N-1)^2}\sigma_a^2\mu_{|\xi|^2} + \frac{(N_0^\BS)^2}{P^2(N-1)^2}.
    \label{eq_89}
\end{align}
The quantities $\E{a_i\delta_i^*|\delta_t|^2}$ in~\eqref{eq_88} and $\E{|\delta_i|^2|\delta_t|^2}$ in~\eqref{eq_89} are ignored in our approximation since they make a negligible contribution to the result.

We approximate the distribution of $f_\Re$ and $f_\Im$ as Gaussian with the above approximate means and variances, and so $\mrm{BER}_{\mrm{DD}2}$ can be obtained in the same manner as in~\eqref{eq_PD_pdf}~--~\eqref{eq_BER_PD}.



\begin{figure}[t!]
    \centering
    \includegraphics[width=0.95\linewidth]{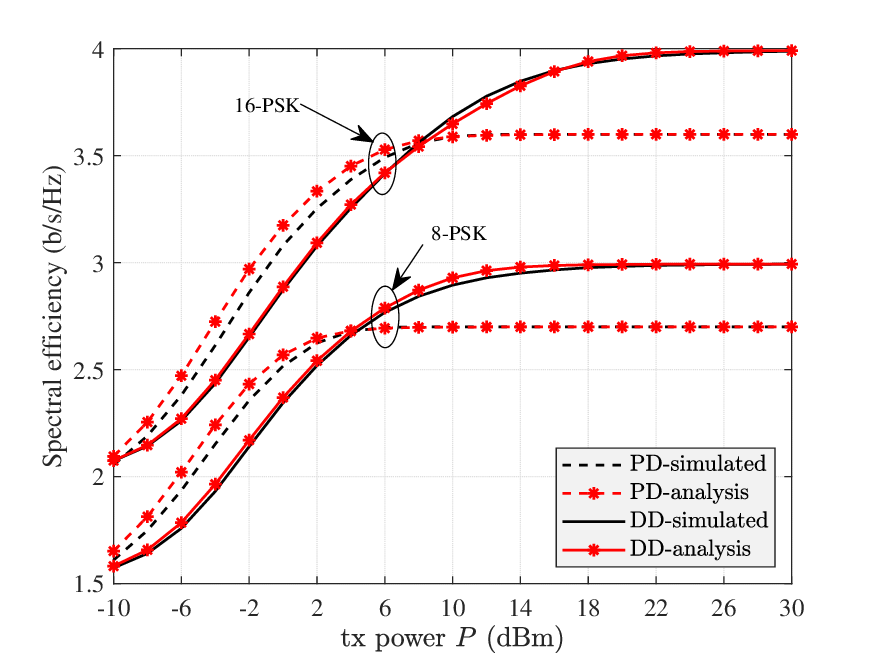}
    \caption{SE comparison with $K = M = 1$ and $N = 50$.}
    \label{fig:SE_comparison}
\end{figure}
\section{Numerical Results}
\label{sec_numerical_results}
In this section, we present various numerical results to verify our SE analysis and to show the benefits of the proposed DD channel estimation framework. We use a general channel model with $\hd{k} = \sqrt{\beta^{\mrm{UB}}_{k}}\mbf{\tilde{h}}_{\mrm{d},k}$, $\mbf{g}_k = \sqrt{\beta^{\mrm{UR}}_{k}}\mbf{\tilde{g}}_k$, $\mbf{H} = \sqrt{\beta^{\mrm{RB}}}\mbf{\tilde{H}}$ where $\mbf{\tilde{h}}_{\mrm{d},k} \sim \mca{CN}(\mbf{0},\bsm{\Sigma}^{\mrm{UB}}_k)$, $\mbf{\tilde{g}}_k \sim \mca{CN}(\mbf{0},\bsm{\Sigma}^{\mrm{UR}}_k)$, and $\mbf{\tilde{H}} = (\bsm{\Sigma}^{\mrm{B}})^{1/2}\mbf{\Bar{H}}(\bsm{\Sigma}^{\mrm{R}})^{1/2}$. The elements of $\mbf{\Bar{H}}$ are i.i.d. and normally distributed as $\mca{CN}(0,1)$. The large-scale fading coefficients are $\beta^{\mrm{UB}}_{k} = \beta_0 (d^{\mrm{UB}}_k/d_0)^{-\alpha^{\mrm{UB}}}$, $\beta^{\mrm{UR}}_{k} = \beta_0 (d^{\mrm{UR}}_k/d_0)^{-\alpha^{\mrm{UR}}}$, $\beta^{\mrm{RB}} = \beta_0 (d^{\mrm{RB}}/d_0)^{-\alpha^{\mrm{RB}}}$ where $\alpha^{\mrm{UB}}$, $\alpha^{\mrm{UR}}$, $\alpha^{\mrm{RB}}$ are the respective path loss exponents, and $d^{\mrm{UB}}_k$, $d^{\mrm{UR}}_k$, $d^{\mrm{RB}}$ are the respective distances between user-$k$ and the BS, user-$k$ and the RIS, and the RIS and BS. We set $\beta_0 = -20$ dB as the path loss at the reference distance $d_0 = 1$m, $\alpha^{\mrm{UB}} = 4$, $\alpha^{\mrm{UR}} = \alpha^{\mrm{RB}} = 2.2$, and the coherence block length at $\tauc = 500$ symbols. If not specifically stated, the noise power is set to $-169$ dBm/Hz and a bandwidth of 1 MHz is assumed. The number of elements with active receivers at the RIS is taken to be equal to the number of users.

First, we numerically validate our derived analytical SE expressions in Figs.~\ref{fig:SE_comparison} and~\ref{fig:SE_crosspoint_vs_N}. Here, we set both the user-RIS and RIS-BS distances to 100 m. In Fig.~\ref{fig:SE_comparison}, the SE is evaluated for $50$ RIS elements versus the transmit power for the cases of 8-PSK and 16-PSK modulation. It is observed that when the transmit power is low, i.e., at low signal-to-noise ratios (SNR), the SE of the PD approach is higher than for DD, but the situation reverses as the SNR grows. This is because the channel estimation performance of the DD approach strongly depends on the data detection performance at the RIS; at higher SNRs the data symbols detected by the RIS are more reliable and this results in better channel estimation and higher spectral efficiency. There is a critical SNR point at which the DD approach begins to perform better than PD, which in this scenario is about 4-dBm and 8-dBm for 8-PSK and 16-PSK, respectively. It can also be seen from Fig.~\ref{fig:SE_comparison} that our analytical SE approximations match well with the numerical results and accurately predict the performance crossover point. Thus the analytical SE can be used in the system design to determine the crossing point and decide whether the PD or DD approach should be used.

We evaluate the SE of the PD and DD approaches as the number of RIS elements $N$ increases in Fig.~\ref{fig:SE_crosspoint_vs_N}, where the transmit power is fixed at 5-dBm. It is interesting to observe that increasing the number of RIS elements can actually lead to a reduction in SE for the PD framework, since the pilot overhead of the PD approach grows proportionally with $N$ leading to a reduction in the number of time slots available for data transmission. On the other hand, the pilot overhead of the DD framework does not depend on the number of RIS elements, and thus the DD approach does not suffer from SE reduction as $N$ increases. Again, our analysis accurately predicts the performance crossover point, which is an important factor for the system design. 
\begin{figure}
    \centering
    \includegraphics[width=0.95\linewidth]{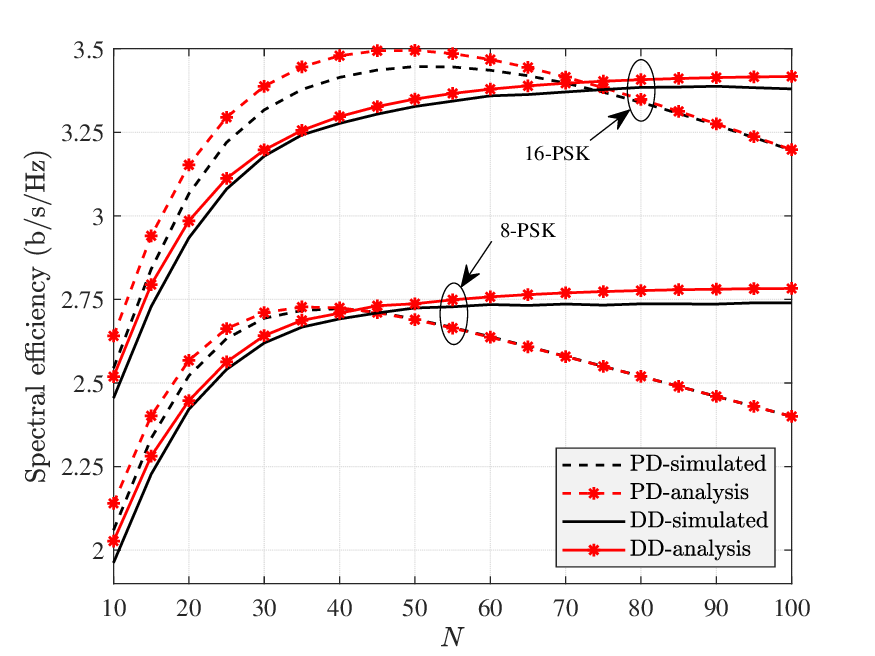}
    \caption{SE comparison with $K = M = 1$, $N$ varies, and $P = 5$ dBm.}
    \label{fig:SE_crosspoint_vs_N}
\end{figure}
\begin{figure*}[t!]
    \centering
    \begin{subfigure}[t]{0.45\linewidth}
        \centering
        \includegraphics[width=\linewidth]{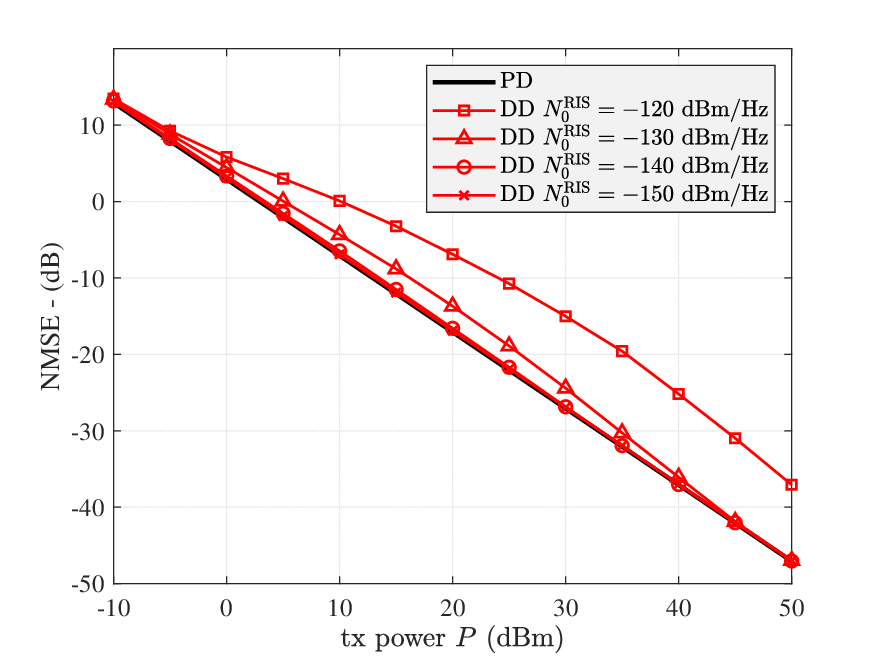}
        \caption{Channel estimation}
        \label{fig_CE_K1}
    \end{subfigure}~
    \begin{subfigure}[t]{0.45\linewidth}
        \centering
        \includegraphics[width=\linewidth]{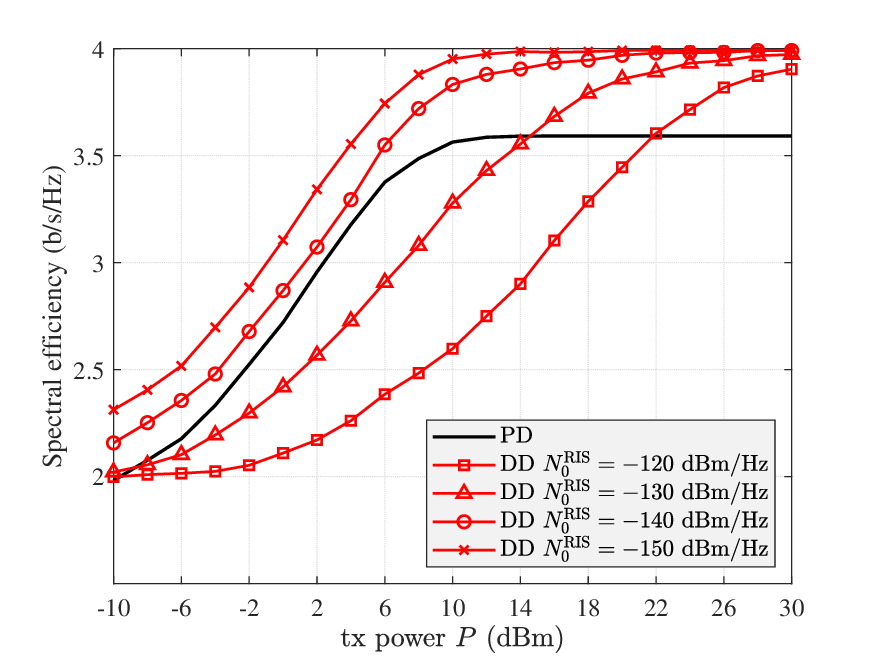}
        \caption{Spectral efficiency}
        \label{fig_SE_K1}
    \end{subfigure}
    \caption{Performance comparison for $K=1$, $M = 8$, $N = 50$, $\rho^\mcA_i = 0.5$, and 16-PSK.}
    \label{fig_CE_SE_K1_comparison}
\end{figure*}

\begin{figure}
    \centering
    \includegraphics[width=0.95\linewidth]{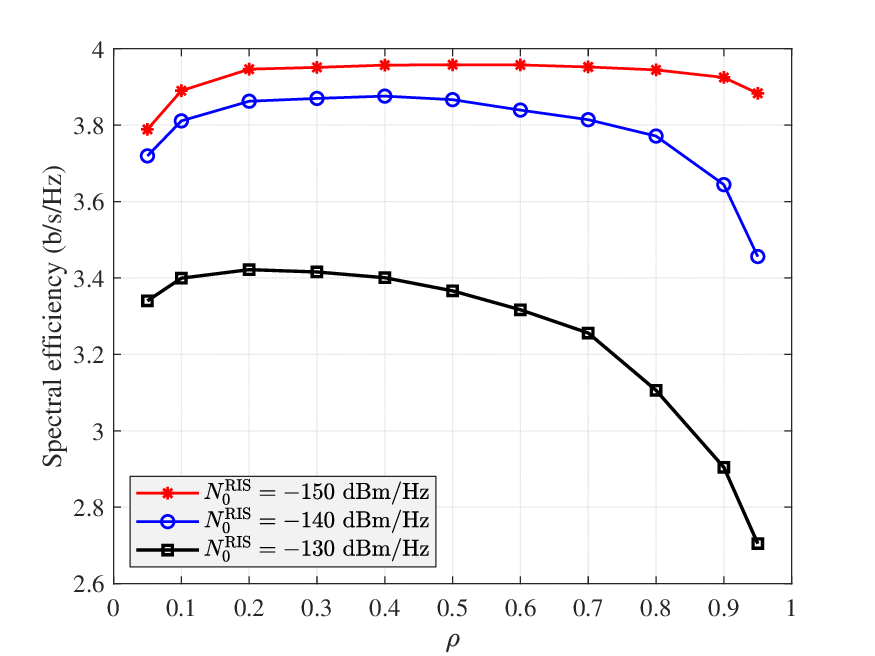}
    \caption{Spectral efficiency of the DD approach as $\rho^\mcA_i$ varies, $P = 10$ dBm, $K = 1$, $M=8$, $N=50$, and 16-PSK signalling.}
    \label{fig:SE_vs_rho}
\end{figure}
In Figs.~\ref{fig_CE_SE_K1_comparison} and~\ref{fig:SE_vs_rho} we also consider a single user scenario but the BS is equipped with multiple antennas. After the channel estimation stage, the phase shift vector $\bsm{\phi}$ of the RIS is optimized to maximize the effective channel strength $\|\mbf{\hat{h}}_{\mrm{d},1} + \mbf{\hat{A}}_1\bsm{\phi}\|^2$, which is solved by semi-definite relaxation (SDR). Fig.~\ref{fig_CE_SE_K1_comparison} shows the channel estimation and spectral efficiency performance for $M=8$, $N=50$, $\rho^\mcA_i = 0.5$, and 16-PSK signalling with different noise power levels and user transmit powers. The normalized mean-squared error (NMSE) is computed as $\E{\|\mbf{\hat{H}}_{\mrm{c},1} - \Hck{1}\|_{\mrm{F}}^2/\|\Hck{1}\|_{\mrm{F}}^2}$. The results in Fig.~\ref{fig_CE_K1} show that the PD method achieves a better channel estimate than DD, but this gain is offset by the increased training overhead for either higher transmit power or a lower noise figure when the DD method can reliably decode the data at the RIS.

Fig.~\ref{fig:SE_vs_rho} illustrates that there is a trade-off in the choice of the fraction of the incident power $\rho^\mcA_i$ that is reflected by the RIS elements with active receivers. A larger $\rho^\mcA_i$ means more signal power is reflected and less is sensed by the RIS. When the amount of signal power sensed by the RIS is too small, the noise at the RIS may dominate the received signal and cause data detection errors, which in turn leads to lower channel estimation accuracy and SE. One the other hand, if the amount of signal power sensed by the RIS is large so as to efficiently recover the data symbols at the RIS, the signals reflected from the RIS to the BS will be weaker, which can lead to less accurate channel estimation at the BS and a reduction in  SE as well. The trade-off is not too serious to handle for small noise levels, but becomes more important as the SNR decreases. For the cases considered in this example, a relatively small value such as $\rho^\mcA_i=0.2$ appears to provide the best system performance.

To study the case of multiple users, we position the RIS and BS at the locations $(x,y)=(50,50)$ and $(x,y)=(100,0)$, respectively, and we locate 
the users randomly within a square whose side length is $20$m and is centered at the origin. Simulation results for a scenario with $K=4$, $M = 8$, $N = 200$, $\rho^\mcA_i = 0.5$, and 16-PSK signalling are given in Fig.~\ref{fig:SE_multiuser}. In the sub-phase 2b, we employ the conventional zero-forcing (ZF) detector for recovering data symbols at the RIS. To configure the RIS phase shift after the channel estimation stage, we find the $\bsm{\phi}$ that maximizes the minimum signal-to-interference plus noise ration (SINR) using the SDR approach as in~\cite{Wu2019IRS}. It is seen that while the typical user and the other users have the same SE for the PD approach, the SE of the typical user is much higher than the SE of the other users when the DD approach is used since only the typical user transmits during the $N$ time slots of phase-1, while all users transmit data in phase-2. This creates a fairness issue that can be addressed in a number of ways. For example, the red curve shows the result obtained by rotating the role of the typical user among all the users over different coherence blocks. In this approach, the average SE of all the users will be the same, and an improvement compared with the unbalanced case is obtained. In particular, the fair DD approach yields approximately a 60\% in SE performance compared to the PD approach.
\begin{figure}
    \centering
    \includegraphics[width=0.95\linewidth]{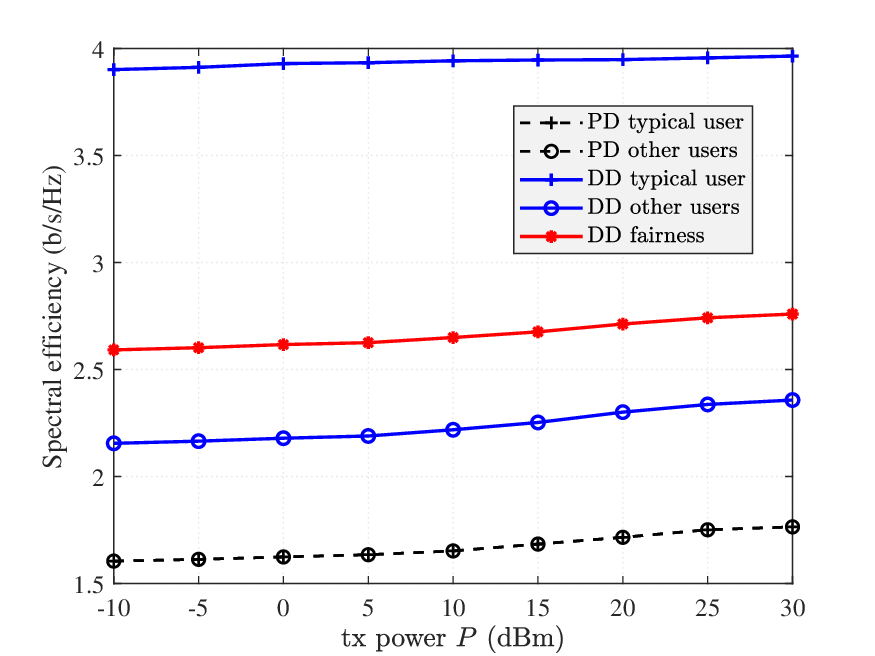}
    \caption{Spectral efficiency comparison with $K=4$, $M = 8$, $N = 200$, $\rho^\mcA_i = 0.5$, and 16-PSK.}
    \label{fig:SE_multiuser}
\end{figure}

Finally, we study the effect of the number of sensing elements $N_\mcA$ on the spectral efficiency in Fig.~\ref{fig:SE_multiuser_vs_NA}. The noise power at the RIS is set to $-120$ dBm/Hz and the transmit power $P$ is 10 dBm. Interestingly, increasing the number of sensing elements $N_\mcA$ only slightly improves the spectral efficiency of the DD approach, indicating that very few sensing elements at the RIS are necessary to achieve the benefit of decision direction. The SE improves more with increasing $N_\mcA$ for the PD approach, since unlike DD, increasing $N_\mcA$ results in a reduction in the pilot overhead of PD.

\section{Conclusion}
\label{sec_conclusion}
In this paper, we have proposed a decision-directed channel estimation framework for general unstructured RIS channel models. It has been shown that with the help of some RIS elements with active receivers, it is possible to accurately estimate the CSI with a pilot overhead only proportional to the number of users and thus significantly improve the spectral efficiency compared to systems with passive RIS arrays. We also performed an intensive spectral efficiency analysis to verify the efficiency of the proposed DD framework. Our analysis takes into account both the channel estimation and data detection errors at both the RIS and the BS, and thus accurately reflects the data detection uncertainty inherent in the decision directed approach.
\begin{figure}
    \centering
    \includegraphics[width=0.95\linewidth]{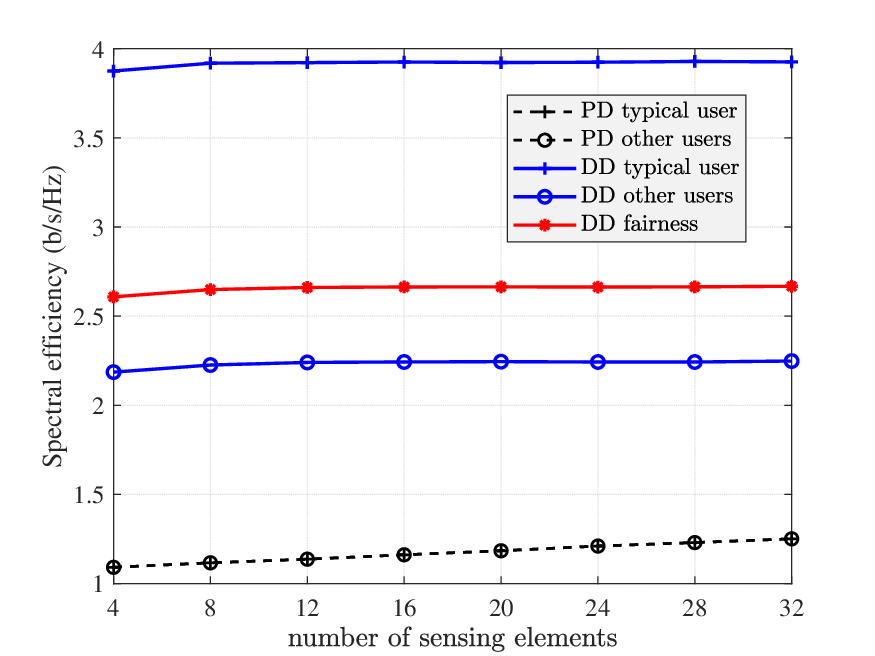}
    \caption{Spectral efficiency comparison versus number of sensing elements $N_\mcA$ with $K=4$, $M = 4$, $N = 200$, $\rho^\mcA_i = 0.5$, 16-PSK signalling, transmit power $P = 20$ dBm.}
    \label{fig:SE_multiuser_vs_NA}
\end{figure}
\ifCLASSOPTIONcaptionsoff
  \newpage
\fi

\appendices
\section{Proof of Theorem~\ref{theorem_1}.}
\label{appendia_1}
The symbol error rate (SER) can be approximated as
\begin{equation}
    \mrm{SER} \approx \mbb{P}[\tilde{y}_\Re\tan\theta - \tilde{y}_\Im \leq 0] + \mbb{P}[\tilde{y}_\Re\tan\theta + \tilde{y}_\Im \leq 0]
\end{equation}
where $\tilde{y}_\Re\tan\theta - \tilde{y}_\Im = 0$ and $\tilde{y}_\Re\tan\theta + \tilde{y}_\Im = 0$ define the rotated decision boundaries. We have  $(\tilde{y}_\Re\tan\theta - \tilde{y}_\Im) \sim \mca{N}(\tilde{\mu},\tilde{\sigma}^2)$ and $(\tilde{y}_\Re\tan\theta + \tilde{y}_\Im) \sim \mca{N}(\tilde{\mu},\tilde{\sigma}^2)$ where
\begin{align*}
    \tilde{\mu} &= \sqrt{P}N\mu_{z_\Re}\tan\theta,\\
    \tilde{\sigma}^2&= \left(PN\sigma_{z_\Re}^2+\frac{N_0}{2}\right)\tan^2\theta + PN\sigma_{z_\Im}^2+\frac{N_0}{2}.
\end{align*}
Therefore,
\begin{align}
    &\mbb{P}[\tilde{y}_\Re\tan\theta - \tilde{y}_\Im \leq 0] = \mbb{P}[\tilde{y}_\Re\tan\theta + \tilde{y}_\Im \leq 0]\notag \\
    &=Q\left(\frac{\sqrt{P}N\mu_{z_\Re} \tan \theta}{\sqrt{\left(PN\sigma_{z_\Re}^2+\frac{N_0}{2}\right)\tan^2\theta + PN\sigma_{z_\Im}^2+\frac{N_0}{2}}}\right),
\end{align}
which means the SER can be approximated as 
\begin{equation}
    \mrm{SER} \approx 2Q\left(\frac{\sqrt{P}N\mu_{z_\Re} \tan \theta}{\sqrt{\left(PN\sigma_{z_\Re}^2+\frac{N_0}{2}\right)\tan^2\theta + PN\sigma_{z_\Im}^2+\frac{N_0}{2}}}\right).
    \label{eq_SER_PD_high_SNR}
\end{equation}

At high SNRs, $\epsilon$ is small, and we have
\begin{align}
    \mu_{z_\Re} &= \E{z_{i,\Re}} \approx \E{|a_i|} = \E{|h_i|}\E{|g_i|} = \frac{\pi}{4}\sigma_a,
    \label{eq_mu_zRe}\\
    \sigma_{z_\Re}^2 &= \mrm{Var}[z_{i,\Re}] = \mbb{E}[z_{i,\Re}^2] - |\mbb{E}[z_{i,\Re}]|^2\notag\\
    &\approx \mbb{E}[|a_i|^2] - \mbb{E}[|a_i|]^2 = \left(1 - \frac{\pi^2}{16}\right)\sigma_a^2.
    \label{eq_sigma2_zRe}
\end{align}
Substituting~\eqref{eq_mu_zRe} and~\eqref{eq_sigma2_zRe} into~\eqref{eq_SER_PD_high_SNR} and using the result that $\mrm{BER} \approx \mrm{SER}/\log_2(D)$ for a Gray code at high SNRs gives us the approximated BER in~\eqref{eq_BER_PD_high_SNR}.

\section{Proof of Lemma~\ref{lemma_1}.}
\label{appendia_2}
We have
\begin{align}
    \E{\xi_t} &= 1 - \E{s_t\hat{s}_t^*}, \label{eq:Exi}\\
    \E{\xi_t^2}&=1-2\E{s_t\hat{s}_t^*} + \E{(s_t\hat{s}_t^*)^2}.\label{eq:Exi2}
\end{align}
Thus, to obtain $\E{\xi_t}$ and $\E{\xi_t^2}$, we need to compute $\E{s_t\hat{s}_t^*}$ and $\E{(s_t\hat{s}_t^*)^2}$, which are given as follows:
\begin{align}
    \E{s_t\hat{s}_t^*} &= \sum_{d=0}^{D-1} p_d^{\DDone}\mca{S}(0)\mca{S}(d)^* \notag\\
    &=p_{0}^{\DDone} - p_{\frac{D}{2}}^{\DDone} + 2\sum_{d=1}^{\frac{D}{2}-1}p_{d}^{\DDone}\cos\left(\frac{2\pi d}{D}\right), \label{eq:Esshat}\\
    \E{(s_t\hat{s}_t^*)^2} &= \sum_{d=0}^{D-1} p_d^{\DDone}(\mca{S}(0)\mca{S}(d)^*)^2 \notag \\
    &= p_{0}^{\DDone} + p_{\frac{D}{2}}^{\DDone} + 2\sum_{d=1}^{\frac{D}{2}-1}p_{d}^{\DDone}\cos\left(\frac{4\pi d}{D}\right). \label{eq:Esshat2}
\end{align}
Substituting~\eqref{eq:Esshat} and~\eqref{eq:Esshat2} into~\eqref{eq:Exi} and~\eqref{eq:Exi2}, we obtain $\mu_\xi$ and $\mu_{\xi^2}$ as in~\eqref{eq:mu_xi} and~\eqref{eq:mu_xi2}, respectively.
The expectation of $\left|\xi_t\right|^2$ is given as
\begin{align}
    \E{\left|\xi_t\right|^2} &= \E{|\hat{s}_t - s_t|^2} =\sum_{{d}=0}^{D-1} \left|\mca{S}({d})-\mca{S}(0)\right|^2\, p_{{d}}^{\DDone} \notag\\
    &= 4p_{\frac{D}{2}}^{\DDone} + 4\sum_{{d}=1}^{\frac{D}{2}-1} \left[1-\cos \left(\frac{2\pi{d}}{D}\right)\right] p_{{d}}^{\DDone} .\label{eq:Eabsxi2}
\end{align}
Note that in~\eqref{eq:Esshat}, \eqref{eq:Esshat2}, and~\eqref{eq:Eabsxi2} we have used the following results: $p_d^{\DDone} = p_{d+\frac{D}{2}}^{\DDone}$, $\mca{S}(0)\mca{S}(d)^* = \mca{S}(0)\mca{S}(d+D/2)^* = \cos(4\pi d/D)$, and $\left|\mca{S}({d})-\mca{S}(0)\right|^2 = \left|\mca{S}({d + D/2})-\mca{S}(0)\right|^2 = 4(1-\cos \left(2\pi{d}/D\right))$ for $d = 1,\,\ldots,\,D/2-1$.

\bibliographystyle{IEEEtran}
\bibliography{ref}

%









\end{document}